\documentclass[a4paper,11pt]{article}
\usepackage{amsmath}
\usepackage{amssymb}
\usepackage{graphicx}
\usepackage{appendix}
\usepackage{color}
\usepackage{hyperref}
\usepackage{times}
\usepackage{fullpage}
\usepackage{mathrsfs}
\usepackage{scrextend}

\newcommand{\remove}[1]{}

\newtheorem{proposition}{Proposition}

\newtheorem{corollary}{Corollary}

\newtheorem{definition}{Definition}

\newtheorem{lemma}{Lemma}

\newtheorem{theorem}{Theorem}

\newtheorem{remark}{Remark}

\newenvironment{proof}{\noindent{\bf Proof\@:}}{\hfill $\Box$\\}

\newenvironment{lemmaproof}[1]{\noindent{\bf Proof of Lemma #1\@}}{\hfill $\Box$\\}

\newenvironment{propositionproof}[1]{\noindent{\bf Proof of Proposition #1\@:}}{\hfill $\Box$\\}

\newcommand{\myshow}[1]{}

\date{\today}
\date{}
\title{Reconstruction/Non-reconstruction Thresholds for \\ Colourings of General Galton-Watson Trees.}

\author{Charilaos Efthymiou\\
College of Computing, Georgia Institute of Technology, Atlanta, USA\\	
{\tt efthymiou@gmail.com}}

\begin{document}
\maketitle

\begin{abstract}
The broadcasting models on trees arise in many contexts such as discrete mathematics, bio\-lo\-gy, 
information theory,  statistical physics and computer science. In this work, we consider the $k$-colouring 
model.    A basic question here is whether the root's assignment affects the distribution  of the 
colourings at the vertices at distance $h$ from the root.  This is the 
so-called  {\em reconstruction  problem}. 
For the case where the underlying tree is  $d$-ary  it is well known that $d/\ln d$ is 
the {\em reconstruction threshold}.  That is, for $k=(1+\epsilon)d/\ln d$ we have non-reconstruction while
for $k=(1-\epsilon)d/\ln d$ we have reconstruction.

Here, we consider the largely unstudied  case where the underlying tree is chosen according to  a predefined distribution.
 In particular,   our focus is on the well-known Galton-Watson trees. This model arises naturally in many contexts, e.g.
the theory of {spin-glasses}  and its applications  on   {random Constraint Satisfaction Problems} (rCSP).   
The aforementioned study  focuses on Galton-Watson trees with offspring  distribution ${\cal B}(n,d/n)$,
i.e. the binomial with parameters $n$ and $d/n$,  where $d$ is fixed.  Here we consider  a {\em broader} version of the problem,  
as we assume  {\em general offspring distribution}, which includes ${\cal B}(n,d/n)$ as a special case.

Our approach relates the corresponding bounds for (non)reconstruction to certain 
{\em concentration properties} of the offspring distribution. This allows to derive  reconstruction
thresholds   for a very wide family of offspring distributions, which 
includes  ${\cal B}(n,d/n)$. 
A very interesting corollary is that  for  distributions with  expected offspring $d$,  we get  
reconstruction threshold   $d/\ln d$ under  {\em weaker concentration}  conditions than  what we have in ${\cal B}(n,d/n)$.
 
Furthermore, our  reconstruction threshold for the random colorings of   Galton-Watson with offspring 
${\cal B}(n,d/n)$,  implies the  reconstruction threshold  for the random colourings of $G(n,d/n)$.

\end{abstract}

\setcounter{page}{1}
\section{Introduction}

The broadcasting models on trees and the closely related reconstruction problem 
are studied in  statistical physics, biology, communication theory, e.g. see \cite{OptPhylogeny,PhaseTransPhylog,IsingTree}.
Our work  is motivated from the  study of {\em  random Constraint Satisfaction Problems} (rCSP) such as 
random graph colouring, random $k$-SAT  etc. 
This is mainly because the models on random trees   capture some of the most fundamental properties 
of the corresponding models on random (hyper)graphs, \cite{LocalWeak,GershMont,MonResTet}.

The most fundamental problem in the study of broadcasting models is to determine the
re\-con\-struc\-tion/non-re\-con\-struc\-tion threshold. I.e. whether the configuration  of the 
root biases the distribution of the configuration of distant vertices.  The transition from 
non-reconstruction to reconstruction can be achieved by adjusting appropriately the 
parameters of the model. Typically, this transition exhibits a {\em threshold behaviour}.

So far, the main focus of the study was to determine the precise location of  this threshold for various 
models when  the underlying graph is a fixed  tree, mostly  regular.  In a lot of applications,
e.g. phylogeny reconstruction, rCSP, usually the underlying  tree is random. 
Motivated by such problems, in this work  we study the  reconstruction problem for the colouring
model
when the underlying tree is chosen according to some predefined probability distribution. 
In particular,  we consider {\em Galton-Watson} trees (GW-trees)  with some  {\em general}  offspring distribution.

The main technical challenge  is to deal with is the so-called 
``effect of high degrees".
That is,  we  expect to have vertices in the tree which are of degree much  higher than the 
expected offspring. 
The deviation from the expected degree is so large that expressing the 
(non)reconstruction bounds in terms of maximum degree leads 
 to highly suboptimal results. 
Similar challenges appear  in   problems in random graphs $G(n,d/n)$  e.g. 
sampling colourings \cite{MyMCMC,mysampling1,mysampling2,FPTASampling}.

It is a folklore conjecture  that when the offspring distribution is ``reasonably" concentrated about its expectation,
 then the  reconstruction  threshold  can be expressed in terms of the expected offspring
of the underlying tree. Somehow, the concentration makes the high degree vertices sufficiently rare, such that
their effect on the phenomenon is negligible. Our aim is to make the  intuitive base of this relation   
{\em rigorous}  by just adopting  the most  generic  assumptions about the offspring distribution.

More specifically,  our result summarizes  as follows: We provide a concentration criterion  for the distributions
over the non-negative integers about the expectation. For a GW-tree with offspring di\-stri\-bu\-tion that satisfies this criterion,  
the transition from non-reconstruction  to reconstruction  exhibits a  threshold behaviour at the 
cri\-ti\-cal point $d/\ln d$, where $d$  is the expected offspring.

Interestingly, the aforementioned concentration criterion  is much  weaker than the standard tail bounds
we  have  for many natural distributions, e.g.   ${\cal B}(n,d/n)$.
On the other hand, when the concentration of the offspring distribution is not sufficiently high to provide thresholds, we still get    
 upper and lower bounds for reconstruction and  non-reconstruction, respectively. These bounds are expressed in terms of  the tails of the offspring 
distribution.

Concluding, let us remark that the reconstruction threshold we get for the random colourings of GW-tree with offspring ${\cal B}(n,d/n)$,
allows to compute the corresponding threshold for the random colourings of $G(n,d/n)$ 
\cite{LocalWeak,GershMont,MonResTet}.  See Section \ref{sec:GWTree2Gnp} for more discussion.

\section{Definitions and Results }

For the sake of brevity, we define the colouring model and the reconstruction problem, first, in
terms of a fixed complete $\Delta$-ary $T$ of height $h$,  where $\Delta, h>0$ are  integers. 
Later we will extend these definitions w.r.t. GW trees.

The broadcasting  models on a tree $T$  are models where  information is sent from the
root over the edges to the leaves.  For some finite set of spins  (colours) ${S}=\{1,2,\ldots, k\}$, a configuration on $T$ is 
an element in $S^T$, i.e. it is an assignment of spins to the vertices of $T$. 
The spin of the root $r$ is chosen according to some initial distribution over $S$.  The information 
propagates along the edges of the tree as follows: There is a $k \times k$ stochastic  matrix 
$M$ such that if the vertex $v$ is assigned spin $i$, then its child $u$ is assigned spin $j$  
with probability $M_{i,j}$.
The  {\em $k$-colouring} model we consider here corresponds to having $M$ such that 
\begin{displaymath}
M_{i,j}=\left \{
\begin{array}{lcl}
\frac{1}{k-1} &\quad & \textrm{for $i\neq j$}\\
0 && \textrm{otherwise.}
\end{array}
\right.
\end{displaymath}
We let $\mu$ be the  {\em uniform distribution} over the $k$-colourings of $T$.
We also refer to $\mu$ as the Gibbs distribution.
Fixing the  spin (colour assignment) at the root of $T$, the  configuration we get after the 
process has finished is distributed as in $\mu$ conditional the spin of the root.

The reconstruction problem can be cast very naturally in terms of the corresponding Gibbs
distribution. 
More specifically,  let $r(T)$ (or $r_T$) denote the root of the tree $T$. Also, let   $L_h(T)$ be 
 the set of vertices at distance $h$ from the root $r(T)$. 
Finally,  we let $\mu^i$ be the distribution $\mu$ conditional that  the spin at $r_T$ is  $i$. 
Reconstructibility is defined as follows: 

\begin{definition}\label{def:Reconstuction}
For any $i,j\in S$ let $||\mu^i -\mu^j||_{L_h}$ denote the total variation distance
of the projections of $\mu^i$ and $\mu^j$ on $L_{h}$.  We say that a model is {\em reconstructible} on 
a tree $T$ if there exists $i,j\in S$ for which
\begin{displaymath}
\lim_{h\to \infty}||\mu^i-\mu^j||_{L_h(T)}>0.
\end{displaymath}
When the above limit is zero for every $i,j$, then we say that the model has
{\em non-reconstruction}.
\end{definition}

\noindent
Non-reconstruction implies, also,  that {\em typical} colourings 
of the vertices at level $h$ of the tree have  a vanishing effect on the distribution 
of the colouring of $r(T)$, as $h$ grows.

For the colouring model on $\Delta$-ary trees it is well-known that  the reconstruction threshold is  $\Delta/ \ln \Delta$, 
 see \cite{TreeNonNaya,InfFlowTrees,SemerjianFreez,SlyRecon}. That is, for any given fixed
 $\epsilon>0$ and sufficiently large $\Delta$, i.e. $\Delta\geq \Delta(\epsilon)$, when 
$k\geq (1+\epsilon)\Delta/\ln \Delta$ we have non-reconstruction while for 
$k\leq (1-\epsilon)\Delta/\ln \Delta$ we have reconstruction.

Rather than considering a fixed tree, here, we consider a Galton Watson tree   (GW-trees) with some
 {\em general} offspring  distribution.  In particular, we let the following:
\begin{definition}\label{def:xi-dist}
Let $ {\xi}$ be a distribution over the non negative integers. We let ${\cal T}_{\xi}$ 
denote a Galton-Watson tree with offspring distribution $\xi$. Also, given some integer $h>0$,
we let ${\cal T}^h_{\xi}$ denote the restriction of  ${\cal T}_{\xi}$ to its first $h$ 
levels\footnote{In other words,  ${\cal T}^h_{\xi}$ is the induced subtree of ${\cal T}_{\xi}$ which contains 
all the vertices within graph distance $h$ from the root.}.
\end{definition}

For the sake of brevity any distribution $\xi$ on the non-negative integers is represented
as a stochastic vector. That is,  for  $Z$ distributed as in  $\xi$ it holds that
$\Pr[Z=i]=\xi(i)$  (or $\xi_i$), for any integer $i\geq 0$. The  notion of reconstruction/non-reconstruction  
from Definition \ref{def:Reconstuction},  extends as follows for Galton-Watson trees:
\begin{definition}\label{def:GWReconstuction}
We say that a model is  
{\em reconstructible} on  ${\cal T}_{\xi}$ if there exists $i,j\in S$ for which
\begin{displaymath}
\lim_{h\to \infty}\mathbb{E}||\mu^i-\mu^j||_{L_h}>0,
\end{displaymath}
where the expectation is w.r.t. the instances of the tree.
When the above limit is zero for every $i,j\in S$, then we say that the model has
{\em non-reconstruction}.
\end{definition}

So as to have a threshold behavior for reconstruction, it is natural to have a certain  kind of parametrization for
the offspring distribution $\xi$. This parametrization   allows to adjust  the expectation from low to high. In what 
follows we assume that we deal with such distribution.

\begin{definition}
Consider ${\cal T}_{\xi}$ for some offspring distribution $\xi$ with expected offspring $d_{\xi}$.
For the $k$-colouring model on ${\cal T}_{\xi}$ we have a {\em reconstruction threshold}
$\theta$ for some function $\theta:\mathbb{R}^+\to\mathbb{R}^+$,  if the following holds: 
For any $\alpha>0$ and $d_{\xi}>d_{\xi}(\alpha)$, 
we have non-reconstruction when $k\geq (1+\alpha)\theta(d_{\xi})$, while 
we have reconstruction when $k\leq (1-\alpha)\theta(d_{\xi})$. 
\end{definition}

\noindent
One of the main results of this work is  to show that we have a threshold behaviour
for the re\-con\-struc\-tion/non-reconstruction transition for the $k$-colourings of ${\cal T}_{\xi}$ 
when $\xi$ is  {\em well concentrated}. The notion of well concentration
is defined as follows:
\begin{definition}\label{def:WellConDist}
A  distribution $\xi$ over the positive  integers with expectation $d_{\xi}$ is
defined to be  ``well concentrated'' if the  following is true:  
There is an absolute constant $c>0$ such that  for any fixed $\gamma>0$,   
$d_{\xi}>d_{\xi}(\gamma)$  and any ${x}\geq (1+\gamma)d_{\xi}$ it holds that
\begin{eqnarray}\label{eq:uniq-condition-threshold}
\sum_{j\geq x}\xi_j \leq x^{-c} \qquad and \qquad \sum_{j\leq (1-\gamma)d_{\xi}}\xi_j \leq (d_{\xi})^{-c}.
\end{eqnarray}
\end{definition}

\noindent
The quantity $c$ is independent of  the distribution $\xi$. We do not compute the exact
value of $c$ but it is   implicit from our derivations.

The following theorem is one of the main results in our work.
\begin{theorem}\label{thrm:Threshold}
Let $\xi$ be a  well concentrated distribution over the non-negative integers. 
Then,   the colouring model on ${\cal T}_{\xi}$ has reconstruction threshold  $d_{\xi}/\ln d_{\xi}$,
where $d_{\xi}$ is the expected offspring.
\end{theorem}
The above theorem follows as a corollary of a more general and more technical result,
 Theorem \ref{thrm:ReconThr}.  This  theorem is  more general  as it covers non-threshold cases, too. 
Given Theorem \ref{thrm:ReconThr}, we provide a proof of  Theorem \ref{thrm:Threshold} in 
Section \ref{sec:thrm:Threshold}.

It is not hard to show that   ${\cal B}(n,d/n)$ is well concentrated. This
follows trivially by just using standard Chernoff bounds (e.g. \cite{RandAlgBook}).
Then,  Theorem \ref{thrm:Threshold} implies   the following corollary.
\begin{corollary}\label{Cor:PoissonGWRes}
Consider  ${\cal T}_{\xi}$ where $\xi$ is the distribution ${\cal B}(n,d/n)$.  
Then,   the colouring model on ${\cal T}_{\xi}$, has  reconstruction
threshold $d/\ln d$.
\end{corollary}
As a matter of fact, it is elementary to verify that  ${\cal B}(n,d/n)$ is, by no means, the less
well concentrated offspring distribution we can have. That is, a distribution with less heavy tails than ${\cal B}(n,d/n)$
can be well concentrated.

\subsection{From Galton-Watson trees to Random Graphs}\label{sec:GWTree2Gnp}

The non-reconstruction phenomenon in rCSP seems to be central  in algorithmic problems. 
In particular, it has been related to the {\em efficiency} of local algorithms which search for 
satisfying solutions. That is,  when we have non-reconstruction, usually there is an efficient 
(simple) local algorithm which finds satisfying assignments efficiently e.g. \cite{BetterKSat,ColorAlg}. On the other hand, in the reconstruction regime
there is no efficient algorithm which finds solutions.  For this reason, the transition from 
non-reconstruction to reconstruction on rCSPs has been attributed the   name ``algorithmic 
barrier" for rCSP\footnote{We should mention that this observation is empirical as there is no corresponding (rigorous) computational 
hardness result.
}, e.g. see \cite{ACO}.

The ingenious, however,  mathematically non-rigorous {\em Cavity Method}, 
introduced by physicists  \cite{cavity,1RSBPaper},  makes  very impressive predictions 
about  the most fundamental properties of  rCSP.
One of the most interesting parts of these predictions involves the Gibbs distribution
and its spatial mixing properties, e.g. the reconstruction problem. 
The Cavity Method predicts that the spatial mixing properties of the Gibbs distribution over
the colouring of $G(n,d/n)$ can be studied by means of the Gibbs distribution of the $k$-colourings
over a Galton-Watson tree with offspring distribution ${\cal B}(n,d/n)$.
That is, choose some vertex $v$ in $G(n,d/n)$ and some fixed radius neighborhood around $v$. 
The  projection of Gibbs distribution  on this neighborhood
 is, somehow, ``similar" to the corresponding Gibbs distribution over the Galton-Watson tree. 
The above line of arguments, led to  conjecture that  the colouring model on a random graph 
$G(n,d/n)$   has the same reconstruction threshold  as that of the GW tree with offspring ${\cal B}(n,d/n)$.

All the above consideration from Cavity method have been studied on a rigorous basis in \cite{LocalWeak,GershMont,MonResTet}.
We have a  quite accurate picture of the relation between the local projection of Gibbs distribution
on $G(n,d/n)$ and the Gibbs distribution on Galton-Watson trees. In particular, we have mathematically 
rigorous arguments which 
imply that indeed the reconstruction thresholds for $G(n,d/n)$ and GW-tree coincide as far as the
colouring model is concerned
\footnote{For more details on the convergence between the  distribution on the GW-tree and $G(n,d/n)$,
 see \cite{LocalWeak}.}.
That is,  Corollary \ref{Cor:PoissonGWRes} 
implies  that, indeed, the reconstruction threshold for the colouring model on $G(n,d/n)$  is $d/\ln d$.

\section{High Level Description}\label{Sec:Approach}

In this section, we give a  high level overview of how do we derive  upper and lower 
bounds for reconstruction and non-reconstruction, respectively. 
Consider an instance of ${\cal T}^h_{\xi}$ for some distribution $\xi$ over the non-negative integers
 and some integer $h>0$.

\begin{remark}
For a set of vertices $\Lambda$ in the tree, we use the 
term ``{\em random colouring of $\Lambda$}''  to indicate  the following 
way of colouring $\Lambda$:
Take a random colouring of the tree and keep only the colouring of the vertices in $\Lambda$.
Also, when we refer to ``{\em typical colourings of vertex set $\Lambda$}'', 
we imply that they are typical w.r.t. the aforementioned distribution.
 \end{remark}

\noindent
Depending on the  tails of $\xi$ we choose appropriate quantities  $\Delta_+$ and $\Delta_-$ such that
 $\Delta_- \leq d_{\xi} \leq \Delta_+$.  Given these two quantities  we show that we have  non-reconstruction  
for  $k\geq (1+\alpha)\Delta_+/\ln \Delta_+$ and we have reconstruction for  $k\leq (1-\alpha)\Delta_-/\ln \Delta_-$, 
for the colouring model on ${\cal T}^h_{\xi}$, where $\alpha>0$ is fixed. 
We show (non)reconstruction by arguing about the  structure of ${\cal T}^h_{\xi}$.

\paragraph{Non Reconstruction.}
First, we focus on  non-reconstruction. Given $\Delta_+$, we  define a set of structural specifications such  
that if ${\cal T}^h_{\xi}$ satisfies them, then we have non-reconstruction for $k\geq (1+\alpha)\Delta_{+}/\ln \Delta_+$. 
We should  consider  $\Delta_+$ to be a  parameter for the specifications.

In particular, given $\Delta_+$,  we  introduce the notion of {\em mixing}  vertex. Roughly speaking, 
a vertex $v\in {\cal T}^h_{\xi}$ is mixing if the following is true:  A typical $k$-colouring 
of the vertices at level $h$   (e.g. Remark 1) does not bias the colouring of $v$ by too 
much when $k\geq (1+\alpha)\Delta_+/\ln \Delta_+$.  A vertex is biased if it is  forced  to 
choose from a relatively  small  set of colours. Perhaps a simple example of a vertex $u$ 
{\em not} being mixing is when the subtree rooted at  $u$ has minimum degree much larger
than  $\Delta_+$.

An  inductive definition  of a mixing vertex, roughly, is as follows: 
A non leaf vertex $v$ is mixing  if the number of its  children is at most $\Delta_+$ while no more
than   $o(\Delta_+)$ of its children  are non-mixing vertices.
 We consider the leaves of the tree  to be mixing vertices, by default.

Furthermore, our specifications require that the mixing vertices are {\em sufficiently many} and {\em well spread} in the tree.
To be more  specific, we want  the following:  For every path from the root of ${\cal T}^h_{\xi}$ to 
the vertices  at level $h$   a sufficiently large fraction  of the  vertices is  mixing. 
Additionally,  we would like that the number of vertices at level $h$ should not deviate significantly from
their expectation.

Then, we  argue  that  non-reconstruction holds for the colouring  model on any, {\em arbitrary}, instance of 
${\cal T}^h_{\xi}$ which satisfies the aforementioned specifications when $k\geq (1+\alpha)\Delta_+/\ln \Delta_+$.
The choice of $\Delta_+$ is the smallest possible that guarantees   that ${\cal T}^h_{\xi}$ satisfies the  structural specifications
 with probability that tends to 1 as $h\to\infty$.

 For showing non-reconstruction, given a fixed tree of the desired structure, we use an idea introduced 
in \cite{BasicIdeaPaper}.   The authors there show non-reconstruction by upper bounding  appropriately the second moment 
of  a quantity called ``magnetization of the root".
This approach has turned out to be  quite popular for showing non-reconstruction bounds
for various models on fixed trees e.g. \cite{TreeNonNaya,SlyRecon,HCTreeNonRecon,BasicIdeaPaper}.
Additionally to \cite{BasicIdeaPaper}, our approach builds on the very elegant combinatorial formalization
from \cite{TreeNonNaya}, which 
uses the notion of {\em unbiasing boundary} to deal with the magnetization of the root.

The approach in \cite{TreeNonNaya} shows non-reconstruction by arguing that the typical colourings
of the vertices at level $h$ do not bias the colouring of the vertices in the largest part of the underlying 
(regular) tree.
The additional element here is that the trees we consider are highly non-regular. So as to get a similar effect 
from the colorings at level $h$, we need to argue about the subtree structure of each vertex in the tree. 
At this point we use the specification requirement. 
In other words, the setting we develop here with the mixing vertices somehow allows to  apply the idea of   unbiasing boundaries 
to control   the magnetization of the root of the non-regular trees we deal with.

\paragraph{Reconstruction}
As opposed to non-reconstruction, the reconstruction  bound  is  well known in the special case 
where the offspring distribution  is  ${\cal B}(n,d/n)$, e.g. \cite{MolloyFr,SemerjianFreez}.  
Our approach deviates  from both \cite{MolloyFr,SemerjianFreez} in  that 
it applies to GW-trees with a general offspring distributions, while 
it focuses 
on  the structural  properties of the underlying  tree, i.e. as we do for the
non-reconstruction bound. 

We are  based on the following observation.
Consider some   fixed tree $T$ of height  $h$ and  some integer $k>0$. Take a random $k$-colouring of the 
vertices at level $h$ of that tree. Consider the probability that the colouring at the root of the tree `freezes" 
by that random  $k$-colouring.  The assignment  at the root gets frozen when the colouring
of the vertices at level $h$ specifies uniquely the colouring at the root.  A sufficient condition for reconstruction is that
the probability that the colouring of the root gets frozen is bounded away from zero for any $h>0$.
 The reconstruction bound  for a $\Delta$-ary tree  follows exactly from this argument, i.e. for 
 $k\leq(1-\alpha) \Delta/\ln \Delta$,  the colouring of the root friezes with probability bounded away
 from zero for any $h$, see  \cite{SemerjianFreez,InfFlowTrees}.

Somehow,  the above   arguments imply  that  if  ${\cal T}^{h}_{\xi}$ has a  $(\Delta_-)$-ary
subtree, with the same root as ${\cal T}^h_{\xi}$,  then   we have reconstruction for $k\leq (1-\alpha)\Delta_-/\ln \Delta_-$.
The structural specification we need for reconstruction is that ${\cal T}^{h}_{\xi}$ has
such a subtree with probability that is bounded away from zero for any $h>0$.
Our  choice of   $\Delta_-$ is  the largest possible  that guarantees exactly this
specification for   ${\cal T}^h_{\xi}$.

\begin{remark}
To be more precise, for non-reconstruction the subtree of ${\cal T}^h_{\xi}$ we consider 
is not exactly $\Delta_-$-ary.  The number of children for each non-leaf vertex is  very close to $\Delta_-$.
\end{remark}

\section{Upper and Lower Bounds}\label{sec:UpLoBounds}

\noindent
We start our analysis by focusing on  the  upper and the lower bounds for reconstruction and non-reconstruction,
respectively.
Consider ${\cal T}^h_{\xi}$ and the $k$-colouring model on this tree.   We  define
appropriate quantities  $\Delta_-$ and $\Delta_+$ which depend (mainly) on the  statistics of the 
offspring distribution $\xi$. As far as $\Delta_+$ is concerned,  we have the following:

\begin{definition}\label{def:qDelta+}
Consider  a distribution $\xi$ over the non negative integers with expectation $d_{\xi}$.
Given some fixed $\delta\in (0,1/10)$, we let  $\Delta_+=\Delta_+(\delta) \geq d_{\xi}$ be the minimum integer such 
that the following holds: There is  $q\in [0,3/4)$  and  $\beta\geq 4$, independent of $d_{\xi}$, 
such that 
\begin{eqnarray}
 q\geq \sum_{i>\Delta_+}\xi_i+\Pr\left[{\cal B}(\Delta_+,q)\geq (\Delta_+)^{\delta}\right] \label{eq:non-reconA}
\end{eqnarray}
and 
\begin{eqnarray}\label{eq:}
 \sum_{t> \Delta_+}t\cdot \xi_t  \leq    \exp\left(-2\beta \ln d_{\xi} \right), \qquad \textrm{}\qquad
 \Pr\left[{\cal B}(\Delta_+,q)>(\Delta_+)^{\delta}\right] \leq   \exp\left(-2\beta \ln d_{\xi} \right). \label{eq:non-reconB}
\end{eqnarray}
\end{definition}

\noindent
Given $\xi$ we choose  $\Delta_+$ as described above. 
Then 
we use $\Delta_+$ as a parameter to specify a set of   structural specifications for  trees
(roughly described in Section \ref{Sec:Approach}).
For any instance of ${\cal T}_{\xi}$ which satisfies these specification 
we have non-reconstruction for any $k\geq (1+\alpha)\Delta_+/\ln \Delta_+$.
The relations between  $\Delta_+$ and $\xi$ as specified  in  (\ref{eq:non-reconA}) and (\ref{eq:non-reconB})  
are, essentially, a  list of  requirements  which guarantee that $\Delta_+$ is as close to
$d_{\xi}$ as possible while at the same time  ${\cal T}^{h}_{\xi}$ satisfies the necessary 
structural specifications  with probability that tends to 1 as $h$ grows.

To illustrate the  intuition behind the relations in Definition \ref{def:qDelta+}, perhaps, it worths
focusing on  (\ref{eq:non-reconA}). 
As we  mentioned before,  the specification requires the tree 
has sufficiently many and well-spread mixing vertices. Then, it is  natural to require that
the probability of a vertex in ${\cal T}^h_{\xi}$ to be mixing is sufficiently large  
regardless of the level of the vertex in the tree.
 The requirement in (\ref{eq:non-reconA}) guarantees   that this probability is appropriately
bounded. 

To be more specific, a vertex $v$ is mixing if the number of its children is at most   $\Delta_+$, while at 
most $\Delta^{\delta}$ of them are allowed to me non-mixing ($\delta$  is as in Definition \ref{def:qDelta+}). 
Let $q$ be an upper bound for the probability of
each child  of $v$ to be non-mixing\footnote{The probability of a vertex  being non-mixing depends only on
the subtree rooted at this vertex.}. Using elementary arguments, we get that  the r.h.s. of (\ref{eq:non-reconA}) is an
upper bound for $v$ to be non-mixing. Moreover, if (\ref{eq:non-reconA}) holds, then  clearly  $q$ is an upper bound for
$v$ to be non-mixing, too.  That is, if some vertex at some level $l$ of the tree is non-mixing with probability at most $q$, then
(\ref{eq:non-reconA}) guarantees that for any vertex at level $l-1$ the probability of it being non-mixing has the
same upper bound $q$. This implies that 
 regardless of its  level at the tree, each vertex $v$ is mixing with probability at least $1-q$.
 The range of $q$ we consider in Definition \ref{def:qDelta+} guarantees that the  mixing vertices
 are as specified by the requirements.
For further details is Section \ref{sec:prop:TinAhz}.

As far as $\Delta_-$ is concerned, we have the following.

\begin{definition}\label{def:gDelta-}
Let $\xi$ be a distribution over the non negative integers. 
Given some $\delta\in (0,1/10)$, we let  $\Delta_-=\Delta_-(\delta)\leq d_{\xi}$ be the maximum integer such 
that the following holds: There is   $g\in [0, 3/4)$  such that 
\begin{eqnarray}
g\geq \sum_{i<\Delta_-}\xi_i+\sum_{i\geq \Delta_-}\xi_i\cdot \Pr\left[{\cal B}(i,1-g)<(\Delta_-)-(\Delta_-)^{\delta}\right]. 
\label{eq:recon}
\end{eqnarray}
\end{definition}

\noindent
The arguments  for reconstruction are based on showing that
with sufficiently large probability the following holds for  ${\cal T}^h_{\xi}$: The root of 
 ${\cal T}^h_{\xi}$ has a subtree of height $h$ such that
each non leaf vertex has sufficiently many children, e.g. approximately $\Delta_-$ many.
We will see in Section \ref{sec:thrm:ReconThrB}, that the  condition in
(\ref{eq:recon})  guarantees that the root of  ${\cal T}^h_{\xi}$ has such a subtree 
with probability bounded away from zero, regardless of the height $h$.
Clearly, this is the structural requirement for reconstruction, we described in Section \ref{Sec:Approach}.

The following theorem is the main technical result of our work.
The trees considered in Theorem \ref{thrm:ReconThr} do not 
necessarily have well concentrated offspring distribution $\xi$.

\begin{theorem}\label{thrm:ReconThr}
Let  some fixed $\alpha>0$.  Consider  an instance of ${\cal T}^h_{\xi}$ such that the expected
offspring $d_{\xi}$ is sufficiently large. 
Set $\delta =\min \{ \alpha/2,1/10\}$, i.e. the variable that specifies both $\Delta_+$ and $\Delta_-$.

For $\mu$, the Gibbs distribution over the  $k$-colourings of ${\cal T}^h_{\xi}$
the following is true:
\begin{description}
\item[non-reconstruction:]
For $k=(1+\alpha)\Delta_+/ \ln \Delta_+$  and any $i,j\in [k]$ it holds that
\begin{eqnarray}
\mathbb{E}||\mu^i-\mu^j||_{L_h}
\leq 8k^2(2\Delta_+)^{-0.45 \delta h}. \nonumber
\end{eqnarray}
\item [reconstruction:]
For $k=(1-\alpha)\Delta_-/ \ln \Delta_-$  there are $i,j\in [k]$ such that
\begin{eqnarray}
\mathbb{E}||\mu^i-\mu^j||_{L_h}\geq \frac{1}{4}\left(1- \frac{2}{\log k}\right). \nonumber
\end{eqnarray}
\end{description}
Both of the expectations above are taken
 w.r.t. the tree instances.
\end{theorem}
The proof of Theorem \ref{thrm:ReconThr} appears in two sections. In Section \ref{sec:thrm:ReconThrA}
we present the proof for the non-reconstruction part.  In Section \ref{sec:thrm:ReconThrB} we present the 
proof for the reconstruction part.

Given Theorem \ref{thrm:ReconThr},  it is elementary to show that Theorem \ref{thrm:Threshold} 
holds. I.e. given that the offspring distribution is well concentrated (Definition \ref{def:WellConDist}), 
we to show that $\Delta_-$ and $\Delta_+$ are sufficiently  close to each other.  The derivations are simple and they 
are presented in full detail in Section  \ref{sec:thrm:Threshold}.\\

\noindent
{\bf Notation.}
For any tree $T$ we let $r(T)$ or $r_T$ denote its root. Let $L_h(T)$ denote the set of 
vertices at graph distance  $h$ from $r(T)$.   
For every vertex $v \in T$,  we define $\tilde{T}_v$  the subtree of $T$ as follows:
Delete the edge between $v$ and its parent in $T$. Then $\tilde{T}_v$ is the connected
component that contains $v$. We use the 
convention that $r(\tilde{T}_v)=v$.

We use capital letter of the Latin alphabet to indicate random variables
which are colourings of the tree $T$, e.g. $X$, $Y$, etc.  We use small letter of 
the greek alphabet to indicate fixed colourings, e.g. $\sigma,\tau$, etc. 
We use the notation $\sigma_{\Lambda}$ or $X(\Lambda)$ do  indicate that the vertices 
in $\Lambda$ have a colour assignment  specified by the colouring $\sigma$ or $X$,
respectively.

Given a tree $T$, we let $\mu$ denote the Gibbs distribution for its $k$-colourings.
Usually we consider $\mu$ under certain boundary conditions, i.e. given some $
\Lambda\subset T$,
and some $k$-colouring of $T$, $\sigma$, we need to consider the Gibbs distribution
where the vertices in $\Lambda$ have fixed colouring $\sigma_{\Lambda}$. For this case
we denote the Gibbs distribution $\mu^{\sigma_{\Lambda}}$.  
For $\Xi\subseteq T$ we let $\mu_{\Xi}$ denote the {\em marginal} of
the Gibbs distribution for the vertices in $\Xi$.  We denote marginals over
the vertex set $\Xi$ of a Gibbs distribution with boundary $\sigma_{\Lambda}$ 
in the natural way,  i.e.  $\mu^{\sigma_{\Lambda}}_{\Xi}$.

\section{Proof of Theorem \ref{thrm:ReconThr} - Non Reconstruction}\label{sec:thrm:ReconThrA}

First, consider a fixed tree $T$ of height $h$ and we let $L=L_h(T)$.
From \cite{Beating2ndEigen} we have that 
\begin{eqnarray}\label{eq:DualRecProb}
||\mu^{i}-\mu||_{r_T}\leq k  \cdot \sum_{\sigma(L)\in[k]^L} \mu_L(\sigma_L) \cdot ||\mu^{\sigma(L)} -\mu||_{r_T}.
\end{eqnarray}
Furthermore,   from the definition of the total variation distance we have that
\begin{eqnarray}
  \sum_{\sigma(L)\in[k]^L} \mu_L(\sigma_L) \cdot ||\mu^{\sigma(L)} -\mu||_{r_T}  &=&
 \frac12 \sum_{\sigma(L)\in[k]^L} \mu_L(\sigma_L) \cdot\sum_{c\in [k]} \left |\mu^{\sigma(L)}_{r_T}(c) -1/k\right|  \nonumber \\
&=&\frac12\sum_{c\in [k]}  \sum_{\sigma(L)\in[k]^L} \mu_L(\sigma_L) \cdot \left |\mu^{\sigma(L)}_{r_T}(c) -1/k\right|.  \label{eq:Red2IndColours}
\end{eqnarray}

\noindent
The quantity $\left |\mu^{\sigma(L)}_{r(T)}(c) -1/k\right|$, is usually called   {\em magnetization of the root $r(T)$}, 
e.g. see  \cite{MagnetizationRef}.  The inner sum 
is  the average magnetization at the root, w.r.t. boundaries at the set $L$. 
We bound this average magnetization by using the following standard result.

\begin{proposition}\label{prop:reduction}
Consider a fixed tree $T$ of height $h$  and some integer $k>0$.   For every $c\in [k]$ the following 
is true:  Let $X$ be a random $k$-colouring  of  $T$ conditional  that $X(r_T)=c$.  It holds that
\begin{eqnarray}\label{eq:thrm:reduction}
  \sum_{\sigma(L)\in[k]^L} \mu_L(\sigma(L)) \cdot
\left |\mu^{\sigma(L)}_{r_T}(c) -1/k\right|  \leq   
 \sqrt{\frac{1}{k} \cdot 
 \left|\left| \mu^{X_L}(\cdot)-\mu^{Z^q_L}(\cdot) \right|\right|_{\{r_T\}}  
},
\end{eqnarray}
where $Z^q$ is random colouring of $T$ conditional that $Z^q(r_T)=q$,
where $q$   maximizes the r.h.s. of (\ref{eq:thrm:reduction}).
\end{proposition}
Our  proof of Proposition \ref{prop:reduction}, which is very similar to the proof of Lemma 1 in 
\cite{BasicIdeaPaper},  appears in Section \ref{sec:prop:reduction}.

The quantity on the r.h.s. of (\ref{eq:thrm:reduction}) is
a deterministic one, i.e. it depends only the  tree $T, c$ and $k$. 
We let 
$$\mathbb{G}_{c,k}(T)= \left|\left| \mu^{X_L}(\cdot)-\mu^{Z^q_L}(\cdot) \right|\right|_{\{r_T\}}  .$$

\noindent
Consider ${\cal T}^h_{\xi}$ as in the statement of Theorem \ref{thrm:ReconThr}.
The quantity  $\mathbb{G}_{c,k}({\cal T}^h_{\xi})$  is a random variable.  
In the light of  (\ref{eq:Red2IndColours}), (\ref{eq:DualRecProb}) and Proposition \ref{prop:reduction},
it suffices to  show that  $\mathbb{E}\left[\mathbb{G}_{c,k}({\cal T}^h_{\xi}) \right]$ tends to 
zero with $h$ sufficiently fast, for any $c\in [k]$ .

\begin{definition}[Mixing Root]\label{def:mixing-root}
Let  $\Delta_+$ and $\delta$ be as in the statement of Theorem \ref{thrm:ReconThr}.
For a tree $T$ of height $h$, its  root  is called mixing if the following holds:
When  $h=0$, then $r(T)$ is  mixing, by default. When $h>0$, $r(T)$ is  mixing  if and only if 
{\tt deg}$ (r_T) \leq  \Delta_+$ and there are at  most $(\Delta_+)^{\delta}$ many  vertices $v$ children 
of $r(T)$ such that $\tilde{T}_v$ does not have a mixing root.
\end{definition}

\begin{definition}
Given $\zeta\in [0,1]$ and some integer $t>0$,  we let ${\cal A}_{t,\zeta}$ denote the set of trees $T$ of height at most $t$ such that the following holds:  
Every  path $\cal P$ of length $h$  from $r(T)$ to $L_t(T)$  contains at least $(1-\zeta)t$  vertices $v$  
such that $\tilde{T}_v$ has a mixing root.  
\end{definition}

\noindent
Before presenting our next  result, we need to do the following remad. In Definition \ref{def:qDelta+},
given $\xi$ and $\delta$, among others the following inequality should hold for $\Delta_+$, 
 \begin{eqnarray}
 \sum_{t\geq \Delta_+}t\cdot \xi_t< \exp\left(-2\beta \ln d_{\xi}\right), 
\nonumber
 \end{eqnarray}
where $\beta\geq 4$. 
Given $\Delta_+$ and $\xi$ the exact value of the parameter $\beta$ is already specified. That is, 
when we define $\Delta_+$ and $\xi$, the value of  $\beta$ is implicit.

\begin{proposition}\label{prop:TinAhz}
Assume that the distribution $\xi$, $\delta$, $\Delta_+$ are as defined in the statement of Theorem 
\ref{thrm:ReconThr}. Let ${\cal C}=\beta \ln d_{\xi}$.
Also, let $\zeta \in(0,1)$ and $\theta=\theta(\zeta)>1$  be such that $(1-\zeta)\theta<1$ and $\beta(1-\theta)<-1$.  
Then, for every $h\geq 1$ it holds that
\begin{eqnarray}
\Pr[{\cal T}^h_ {\xi}\in {\cal A}_{h, \zeta }]\geq 
1-\exp \left[ -(1-\theta(1-\zeta) ){\cal C}\cdot h\right].  \nonumber
\end{eqnarray}
\end{proposition}
The proof of Proposition \ref{prop:TinAhz} appears in Section \ref{sec:prop:TinAhz}.

\begin{theorem}\label{thrm:Basic}
Let $\xi, \delta, \Delta_+$  and $\alpha$ be as in the statement of Theorem \ref{thrm:ReconThr}.
Also, let $\zeta\in (0,1)$ and let the integer $h\geq 1$.
For $k=(1+\alpha)\Delta_+/\ln \Delta_+$, it holds that 
\begin{eqnarray}
 \mathbb{E} \left [\left . 
\mathbb{G}\left({\cal T}^h_{\xi}\right) \right| {\cal T}^h_{ {\xi}}\in {\cal A}_{h,\zeta}
 \right]
\leq
 \frac{4 (2\Delta_+)^{-0.9(3/4-\zeta)\delta h}}{\Pr[{\cal T}^h_{ {\xi}}\in {\cal A}_{h,\zeta}]}.
   \nonumber
\end{eqnarray}
\end{theorem}
The proof of Theorem \ref{thrm:Basic} appears in Section \ref{sec:thrm:Basic}.

Set  $\zeta=1/4$, and $\theta=1.3$, applying Proposition \ref{prop:TinAhz} 
we get that 
\begin{eqnarray}
\Pr[{\cal T}^h_ {\xi}\notin {\cal A}_{h, \zeta }]\leq d^{-0.1h}_{\xi}.\label{eq:from:prop:TinAhz}
\end{eqnarray}
For the same  values of $\zeta,\theta$ as above,   (\ref{eq:from:prop:TinAhz}) with  Theorem \ref{thrm:Basic} 
gives that 
\begin{eqnarray}
 \mathbb{E} \left [\left . \mathbb{G}({\cal T}^h_{\xi})  \right| {\cal T}^h_{ {\xi}}\in {\cal A}_{h,\zeta} \right]\leq
8 (2\Delta_+)^{-0.45 \delta h}.
\label{eq:thrm:Basic}
\end{eqnarray}
Since we always have $0\leq \mathbb{G}(T)\leq 1$, for $\zeta$ and $\theta$ as above,  we get that
\begin{eqnarray}
 \mathbb{E} \left [\mathbb{G}({\cal T}^h_{\xi})  \right]&\leq & \mathbb{E} \left [\left . \mathbb{G}({\cal T}^h_{\xi})  \right| {\cal T}^h_{ {\xi}}\in {\cal A}_{h,1/4} \right]+\Pr\left[{\cal T}^h_{ {\xi}}\notin {\cal A}_{h,1/4}\right ] 
\leq 16 (2\Delta_+)^{-0.45 \delta h}, \nonumber
\end{eqnarray}
where the last inequality follows from  (\ref{eq:from:prop:TinAhz}) and (\ref{eq:thrm:Basic}).
The theorem follows.

\section{Proof of Theorem \ref{thrm:Basic}}\label{sec:thrm:Basic}

Consider first the quantity $\mathbb{G}_{c,k}(T)$,  for  some fixed tree $T$.
Then, it holds that 
\begin{eqnarray}
\mathbb{G}_{c,k}(T)&=& \left|\left| \mu^{X_L}(\cdot)-\mu^{Z^q_L}(\cdot) \right|\right|_{r_T} . 
\end{eqnarray}

\noindent
An important remark from  Proposition \ref{prop:reduction} is that it allows to use any kind
of correlation between the $X,Z^q$. For this reason we assume that $(X,Z^q)$ is distributed 
as in  $\nu^T_{c,q}$. We are going to specify this distribution soon. 
First we get the following result.

\begin{proposition}\label{prop:UpcouplingVsHamming}
Let $\xi, \delta, \Delta_+$  and $\alpha$ be as in the statement of Theorem \ref{thrm:Basic}.
Also  let $0\leq \gamma\leq \delta$.
Then for $k=(1+\alpha)\Delta_+/\ln \Delta_+$, it hold that 
\begin{eqnarray}
 \mathbb{E} \left[\mathbb{G}_{c,k}\left({\cal T}^h_{\xi}\right) \left| {\cal T}^h_{\xi} \in {\cal A}_{h,\zeta} \right. \right ]&\leq &
\frac{1}{\Pr\left[{\cal T}^h_{\xi} \in {\cal A}_{h,\zeta}\right]} 
\left(2 \exp\left(-\frac18(\Delta_+)^{\frac{h/4-1}{2}\delta+\frac{7}{8}\frac{\alpha}{1+\alpha}} \right)
\cdot \mathbb{E}\left[\left|L_h\left({\cal T}^h_{\xi}\right)\right|\right] + \right  . \nonumber \\
&&  \left .  +2(2(\Delta_+)^{-\gamma})^{\left(3/4-\zeta\right)h}\cdot \mathbb{E}[H(X_L,Z^q_L)]  \right).  
\label{eq:prop:UpcouplingVsHamming}
\end{eqnarray}
\end{proposition}
For the above proposition we remark the following: On the r.h.s. of (\ref{eq:prop:UpcouplingVsHamming}) the rightmost expectation term is w.r.t. both the joint distribution 
of $X,Z^q$ and the distribution over the tree ${\cal T}^h_{\xi}$. The rest expectations are w.r.t. the distributions over trees only, i.e.  ${\cal T}^h_{\xi}$.  
The proof of Proposition \ref{prop:UpcouplingVsHamming} appears in Section \ref{sec:prop:UpcouplingVsHamming}.

For showing the theorem we bound appropriately the two expectations on the r.h.s. of  (\ref{eq:prop:UpcouplingVsHamming}).
It is elementary   that 
\begin{eqnarray}
\mathbb{E}\left[\left|L_h\left({\cal T}^h_{\xi}\right)\right|\right]=\left(d_{\xi}\right)^h. \label{eq:Bound1BasicThrm}
\end{eqnarray}

\noindent
For bounding $\mathbb{E}\left [H(X_L, Z^q_L)\right]$ we need to specify a coupling between the random variables $X$ and $Z^q$
which minimizes their expected Hamming distance. Observe that the expected
hamming distance is both w.r.t. the coupling and the randomness of the trees. 

The coupling of $X$ and $Z^q$ we use, can be defined inductively  as follows: We colour the vertices from
the root down to the leaves. For a vertex $v$ whose father $w$ is such that $X(w)=Z^q(w)$ we  
couple $X(v)$ and $Z^q(v)$ identically,  i.e. $X(v)=Z^q(v)$. On the other hand, when $X(w)\neq Z^q(w)$ we  set  $X(v)=Z^q(v)$ unless $X(v)=Z^q(w)$, then we set $Z^q(v)=X(w)$.

Let   $w$ be a vertex in the tree and let $u$ be a child of $w$.
Then, for the coupling above, it holds that
$$
\Pr\left[X(u)\neq Z^q(u)|X(w)\neq Z^q(w)\right]=k^{-1}.
$$
In  ${\cal T}^h_{\xi}$,   the expected number of children per (non-leaf) vertex is
$d_{\xi}$. Then, it is elementary to show that for a disagreeing vertex, the 
expected number of  disagreeing children is $d_{\xi}/k\leq \frac{\ln \Delta_+}{1+\alpha}$, 
since $\Delta_+>d_{\xi}$. 
Furthermore, it holds that 
\begin{eqnarray}\label{eq:Bound2BasicThrm}
\mathbb{E}[H(X_L,Y_L)]\leq \left( \frac{\ln \Delta_+}{(1+\alpha)} \right)^h.
\end{eqnarray}
Observe that the above expectation is w.r.t. {\em both} tree instances and random colourings. 

The theorem follows be combining (\ref{eq:Bound2BasicThrm}),  (\ref{eq:Bound1BasicThrm})  and
 Proposition \ref{prop:UpcouplingVsHamming}.

\section{Proof of Proposition \ref{prop:UpcouplingVsHamming}}
\label{sec:prop:UpcouplingVsHamming}

The previous setting allows to use  ideas based on the notion of 
biasing-unbiasing boundary  (introduced in \cite{TreeNonNaya}) to
prove Proposition \ref{prop:UpcouplingVsHamming}.
To be more precise, the definition of biasing non-biasing boundaries we use 
here  is slightly different than that \cite{TreeNonNaya}, but the approach is
 similar.

\begin{definition}[Non-Biasing Boundary]\label{def:BiasBound}
For  $\alpha,\gamma,\delta,\Delta_+$  as in the statement of Proposition \ref{prop:UpcouplingVsHamming},
we let $k=(1+\alpha)\Delta_+ /\ln \Delta_+$, and let some integer $t\geq 1$.
Consider a tree  $H$  of height $t$ such that $r(H)$ is  mixing.
For  a $k$-colouring of $H$ $\sigma$ we say that $\sigma_{L}$ does not bias the root if the following holds:
\begin{itemize}
 \item if $t=1$, then $\sigma(L_t(G))$ uses all but at least $(\Delta_+)^{\gamma}$ many colours.
 \item if $t>1$, then the following holds: We let  $v_1, \ldots, v_s$ 
are the children of the root of $H$, where  $s\leq \Delta_+$. Also, let 
  $\mathbb{S}\subseteq \{\tilde{H}_{v_1}, \tilde{H}_{v_2},\ldots, \tilde{H}_{v_s}\}$ 
contain  only the subtrees whose roots are mixing.
Then, there are at most $\Delta^{\delta}_+$ many {subtrees $\tilde{H}_{v_i}\in \mathbb{S}$}
such that  $\sigma(L_{t-1}(\tilde{H}_{v_i}))$ biases the root $r(\tilde{H}_{v_i})$.   
\end{itemize}
Also, we let ${\cal U}(T)$ denote the set of all boundary conditions on $L$ which are not biasing.
\end{definition}

\noindent
Note the notion of non-biasing boundary condition makes sense only for trees with mixing roots.

\begin{lemma}\label{lemma:IndBiasCond}
Let $\gamma,\alpha,\Delta_+$ be as in the statement of Proposition \ref{prop:UpcouplingVsHamming}.
Let $k=(1+\alpha)\frac{\Delta_+}{\ln \Delta_+}$, also let some integer $t\geq 1$.
Consider  a fixed tree $T$ of height $t$ and let  $L=L_t(T)$.
For $\sigma$, a  $k$-colouring of $T$, such that $\sigma_L$ is  biasing for the root of $T$
the following is true:
There is at least one $c\in [k]$ such that for $X$, a random $k$-colouring of $T$,
it holds that 
\begin{eqnarray}\nonumber
\Pr[X_{r(T)}=c|X_L=\sigma_L]\geq (\Delta_+)^{-\gamma}.
\end{eqnarray}
\end{lemma}

The proof of Lemma \ref{lemma:IndBiasCond} appears in Section \ref{sec:lemma:IndBiasCond}.

\begin{definition}
Let  $\alpha,\gamma,\delta,\Delta_+,h$  be as in the statement of Proposition \ref{prop:UpcouplingVsHamming}.
Consider a tree  $T$  of height $h$ and let  $L=L_h(T)$.  For every vertex $w\in L$ we define the set of 
boundaries ${\cal U}_{w}\subseteq [k]^{L}$ as follows:   Let ${\cal P}$ denote the path that 
connects $r_T$ and $w$ and we let 
$$
{\cal M}=\left\{v\in {\cal P}: dist(r_T,v)\leq \frac34h, \; \tilde{T}_v \textrm{ has mixing root} \right\}.
$$
Then ${\cal U}_{w}$ contains the boundary conditions on $L$ which do not bias the root of any of the subtrees  $\tilde{T}_v$ where $v \in{\cal M}$.
\end{definition}

\begin{proposition}\label{prop:Delta-Start}
Let  $\alpha,\gamma,\delta,\Delta_+,h,\zeta$  be as in the statement of Proposition \ref{prop:UpcouplingVsHamming}.
Let some fixed tree $T\in {\cal A}_{h,\zeta}$ and  let $L=L_h(T)$.
Consider $\sigma, \tau$ to be two $k$-colourings of  $T$
such that $H(\sigma_L, \tau_L)=1$.
Furthermore, assume that $\sigma(w)\neq \tau(w)$ for some $w\in L$,  while  both $\sigma_L, \tau_L \in {\cal U}_{w}$. 
Then it holds that
\begin{eqnarray}
 ||\mu^{\sigma_L}-\mu^{\tau_L}||_{r(T)}\leq \Delta^*_{\zeta,h}=
(2\Delta^{-\gamma}_+)^{\left(3/4-\zeta\right)h}. \nonumber
\end{eqnarray}
\end{proposition}
The proof of Proposition \ref{prop:Delta-Start} appears in Section \ref{sec:prop:Delta-Start}.

\begin{proposition}\label{prop:prob-UXL} 
Let  $\alpha,\gamma,\delta,\Delta_+,h,\zeta$  be as in the statement of Proposition \ref{prop:UpcouplingVsHamming}.
Consider a fixed tree $T\in {\cal A}_{h,\zeta}$. 
Let $X$ be a random $k$-colouring of $T$.  For $k=(1+\alpha)\Delta_+/\ln \Delta_+$
and any $w\in L_h(T)$ it holds that
\begin{eqnarray}
 \Pr\left[ X_L\notin {\cal U}_{w} \right]\leq
2 \exp\left(-\frac18(\Delta_+)^{\frac{h/4-1}{2}\delta+\frac{7}{8}\frac{\alpha}{1+\alpha}} \right).
\nonumber
\end{eqnarray}
\end{proposition}
The proof of Proposition \ref{prop:prob-UXL} appears in Section \ref{sec:prop:prob-UXL}.\\

\begin{propositionproof}{\ref{prop:UpcouplingVsHamming}}
First, consider some fixed tree $T\in {\cal A}_{h,\zeta}$ and we let $L=L_h(T)$.
Usually we fix a colouring of $L$ and we call it (the colouring) boundary condition. 
We also use the term ``free" boundary to indicate the absence of any boundary condition on $L$
or some of its vertices.

Consider   two colourings of the leaves $\sigma(L)$ and $\tau(L)$.
We let $m$ be the Hamming distance between $\sigma(L)$ and $\tau(L)$, i.e. $m=H(\sigma_L, \tau_L)$.
Let $v_1,\ldots, v_m$ be the vertices in $L$ for which $\sigma_L$ and $\tau_L$ disagree.  
Consider the sequence of boundary conditions $Z_0, \ldots, Z_{2m}\in [k]^{L}$ such that $\sigma_L=Z_1$, $\tau_L=Z_{2m}$  while  the rest of the members are as follows:  
For $i \leq m$, we get $Z_i$ from $Z_{i-1}$ be substituting  the assignment of $v_i$ from $\sigma(v_i)$ to ``free".
Also, for $i\geq m$ we get $Z_{i+1}$ from $Z_{i}$ by substituting $Z(v_{i-m})$ from ``free"
to $\tau(v_{i-m})$.  It is direct that $H(Z_i,Z_{i+1})=1$. 

It holds that
\begin{eqnarray}
 ||\mu^{\sigma_L}-\mu^{\tau_L}||_{r(T)}\leq \sum^{2m-1}_{i=0}||\mu^{Z_i}-\mu^{Z_{i+1}}||_{r(T)}. \label{eq:TrianZ_i}
\end{eqnarray}

\noindent
Also, it is not hard to see that for every $w\in L$ the following is true:
if $\sigma_L\in {\cal U}_{w}$, then $Z_i\in {\cal U}_{w}$ for every $i=1,\ldots, m$. 
Similarly,  if $\tau_L\in {\cal U}_{w}$,  then $Z_i\in {\cal U}_{w}$ for every $i=m,\ldots, 2m$.

Let the event  $\mathbb{U}^{\sigma,\tau}_{v_i}=$ ``$ \sigma_L, \notin {\cal U}_{v_i} \bigcup \tau_L\notin {\cal U}_{v_i}$".
Then it holds that 
\begin{eqnarray}\label{eq:EachTermTriangl}
||\mu^{Z_i}-\mu^{Z_{i+1}}||_{r(T)}\leq \mathbb{I}_{\{ \mathbb{U}_{v_i} \}}+
\left(1-\mathbb{I}_{\{\mathbb{U}_{v_i}\}}\right) \Delta^*_{\zeta,h},
\end{eqnarray}
where $ \Delta^*_{\zeta,h}$ is defined in the statement of Proposition \ref{prop:Delta-Start}.
In words, the above inequality states the following: if at least one of the $\sigma_L,\tau_L$
are not in ${\cal U}_{v_i}$, then  the l.h.s. of (\ref{eq:EachTermTriangl}) is at most 1. On the 
other hand, if both $\sigma_L,\tau_L \in {\cal U}_{v_i}$ then the total variation distance on 
the l.h.s. can be upper bounded by using Proposition \ref{prop:Delta-Start}.

Plugging (\ref{eq:EachTermTriangl}) into (\ref{eq:TrianZ_i}) we have that
\begin{eqnarray}\label{eq:DetermBound}
 ||\mu^{\sigma_L}-\mu^{\tau_L}||_{r(T)} \leq 2\cdot 
\sum_{v\in L_h(T)}\mathbb{I}_{\{\sigma_v\neq\tau_v\}}\cdot 
\left[\mathbb{I}_{\{\mathbb{U}_{v}\}}+\left (1-\mathbb{I}_{\{\mathbb{U}_{v}\}}\right) \cdot \Delta^*_{\zeta,h}\right].
\end{eqnarray}

\noindent
Now, we consider the quantity $\mathbb{G}_{c,k}(T)$, i.e. $\mathbb{G}_{c,k}(T)=||\mu^{X_L}-\mu^{Z^q_L}||_{r(T)} $.
For bounding $\mathbb{G}_{c,k}(T)$ we are going to use (\ref{eq:DetermBound}). That is

\begin{eqnarray}
\mathbb{G}_{c,k}(T)&=& ||\mu^{X_L}-\mu^{Z^q_L}||_{r(T)} 
\leq \sum_{\sigma_L, \tau_L\in [k]^L} 
\Pr\left[X_L=\sigma_L, Z^q_L=\tau_L \right]
\cdot
 ||\mu^{\sigma_L}-\mu^{\tau_L}||_{r(T)}
\nonumber \\
&\leq&2\cdot \sum_{\sigma_L, \tau_L\in [k]^L}
\Pr\left[X_L=\sigma_L, Z^q_L=\tau_L \right]
\cdot
\sum_{v\in L_h(T)}\mathbb{I}_{\{\sigma_v\neq\tau_v\}}\cdot
\left(\mathbb{I}_{\{\mathbb{U}^{\sigma,\tau}_{v}\}}+\left(1-\mathbb{I}_{\{\mathbb{U}^{\sigma,\tau}_{v}\}}\right) \Delta^*_{\zeta,h}\right)
\quad\mbox{[from (\ref{eq:DetermBound})]}\nonumber \\
&\leq & 2\cdot \sum_{v\in L_h(T)}  \left( \Pr\left [X(v)\neq Z^q(v), \mathbb{U}^{X_L,Z^q_L}_{v} \right] + \Pr\left [X(v)\neq Z^q(v) \right]\cdot \Delta^*_{\zeta,h}\right)
\nonumber\\
&\leq &  
2\cdot
\sum_{v\in L_h(T)}
\Pr \left[ \mathbb{U}^{X_L, Z^q_L}_{v} \right] 
+
2\cdot \sum_{v\in L_h(T)}\Pr\left[ X(v)\neq Z^q(v) \right ]\cdot\Delta^*_{\zeta,h}.\nonumber
\end{eqnarray}
Due to symmetry it holds that $\Pr\left[X(L)\notin {\cal U}_{v}\right]=\Pr\left[Z^q(L)\notin {\cal U}_{v}\right]$.
Using this observation and a union bound, the above  inequality implies that
\begin{eqnarray}
\mathbb{G}_{c,k}(T)
&\leq & 4 \sum_{v\in L}\Pr\left[X(L)\notin {\cal U}_{{\cal P}_v}\right]+
\Delta^*_{\zeta,h} \sum_{v\in L}\Pr\left[X(v)\neq Z^q(v) \right]  \nonumber\\
&\leq& 
2 \exp\left(-\frac18(\Delta_+)^{\frac{h/4-1}{2}\delta+\frac{7}{8}\frac{\alpha}{1+\alpha}} \right)
\cdot |L_h(T)| +2\Delta^*_{\zeta,h}\cdot \mathbb{E}_{\nu_{c,q}}[H(X_L, Z^q_L)], \nonumber
\end{eqnarray}
where in the last inequality we used Proposition \ref{prop:prob-UXL} to bound $\Pr\left[X(L)\notin {\cal U}_{{\cal P}_v}\right]$ . 
$\mathbb{E}_{\nu_{c,q}}[H(X(L),Z^q(L))]$ is the expected Hamming distance between 
$X_L$ and $Z^q_L$ and depends only on the joint distribution of $X, Z^q$,  which is denoted as $\nu_{c,q}$. 

The proposition follows by averaging over ${\cal T}^h_{\xi}$, conditional that we have
a tree in ${\cal A}_{h,\zeta}$, that is
\begin{eqnarray}
 \mathbb{E} \left[\mathbb{G}_{c,k}\left({\cal T}^h_{\xi}\right) \left| {\cal T}^h_{\xi} \in {\cal A}_{h,\zeta} \right. \right ]&\leq &
\frac{1}{\Pr\left[{\cal T}^h_{\xi} \in {\cal A}_{h,\zeta}\right]} 
\left(2 \exp\left(-\frac18(\Delta_+)^{\frac{h/4-1}{2}\delta+\frac{7}{8}\frac{\alpha}{1+\alpha}} \right)
\cdot \mathbb{E}\left[\left|L_h\left({\cal T}^h_{\xi}\right)\right|\right] + \right  . \nonumber \\
&&  \left .  +2(2\Delta_+^{-\gamma})^{\left(3/4-\zeta\right)h}\cdot \mathbb{E}[H(X_L,Z^q_L)]  \right). 
\nonumber
\end{eqnarray}
The rightmost expectation term  is w.r.t. both $\nu_{c,q}$ and the distribution of
random trees ${\cal T}^h_{\xi}$. 
In the above derivations we used the following, easy to derive, inequality
\begin{eqnarray}
 \mathbb{E} \left[f\left({\cal T}^h_{\xi}\right) \left| {\cal T}^h_{\xi} \in {\cal A}_{h,\zeta} \right. \right ]\leq
{ \mathbb{E} \left[f\left({\cal T}^h_{\xi}\right)\right ]}/{\Pr\left[ {\cal T}^h_{\xi} \in {\cal A}_{h,\zeta} \right]},\nonumber
\end{eqnarray}
where $f$ is any non-negative functions on the support of  the distribution ${\cal T}^h_{\xi}$.
The proposition follows.
\end{propositionproof}

\section{Proof of Proposition \ref{prop:Delta-Start}}\label{sec:prop:Delta-Start}

For showing Proposition \ref{prop:Delta-Start} we use coupling.  
The  coupling is standard and it has been used in different contexts, e.g.   \cite{old-GnpSampling, MyMCMC}.

Not at that we have  exactly one disagreement only on some  vertex $w\in L$
in the tree $T$. So as to bound 
$ ||\mu^{\sigma_L}-\mu^{\tau_L}||_{r(T)}$ we take two $k$-colourings of $T$, 
$X$ and $Y$ distributed as in $\mu^{\sigma_L}, \mu^{\tau_L}$ respectively.
We are going to couple $X,Y$ and use the fact that 
\begin{eqnarray}\label{eq:CouplLem}
||\mu^{\sigma_L}-\mu^{\tau_L}||_{r(T)}\leq \Pr[X(r_T)\neq Y(r_T)].
\end{eqnarray}
The coupling of the two random variables is done in a step-wise fashion
moving away from the disagreeing  vertex $w$.  In particular what is of our
interest is the vertices on the path ${\cal P}$ that connects $w$ with $r_T$, i.e.
${\cal P}=v_0,v_1,\ldots v_h$ where $v_0=w$ and $v_h=r_T$.
We couple $X, Y$ by considering the pairs $(X(v_i), Y(v_i))$, for $i=1,\ldots, h$.

If for some $j\in [h]$ we have that $X(v_j)=Y(v_j)$, then we can couple the
remaining vertices in ${\cal P}$ identically, i.e. for every $i>j$ we have $X(v_i)=Y(v_i)$.
Clearly this holds due to the fact that the underlying graph is a tree. Once we have 
$X(v_j)=Y(v_j)$ there is no alternative path for the disagreement to propagate
to the pairs $X(v_i),Y(v_i)$ for any   $i>j$.

On the other hand, consider the case that $X(v_j) \neq Y(v_j)$, for some $h/4 \leq j \leq h$.
We need to bound the probability that $X(v_{j+1}) \neq Y(v_{j+1})$ in the coupling.
For this we consider two cases, depending on whether the tree $\tilde{T}_{v_{j+1}}$
has a mixing root or not. We show that it holds that 
\begin{eqnarray}\label{eq:DisPercProbs}
\Pr\left[X(v_{j+1}) \neq Y(v_{j+1})| X(v_{j}) \neq Y(v_{j}) \right]\leq
\left\{
\begin{array}{lcl}
2\Delta^{-\gamma}_+ &\quad& \textrm{if $\tilde{T}_{v_{j+1}}$ has mixing root} \\ \vspace{-.3cm} \\
1 &&\textrm{otherwise}.
\end{array}
\right .
\end{eqnarray}
Once we show that indeed the above bounds hold,  it is a matter of  straightforward calculations
to show that  the proposition. In particular,  we use  (\ref{eq:CouplLem}) and the trivial bound  that 
\begin{eqnarray}\nonumber
||\mu^{\sigma_L}-\mu^{\tau_L}||_{r(T)}\leq \Pr[X(r_T)\neq Y(r_T)] \leq \prod_{i=h/4}^h \Pr\left[X(v_{i}) \neq Y(v_{i})| X(v_{i-1}) \neq Y(v_{i-1}) \right].
\end{eqnarray}
The probabilities on the r.h.s.  are substituted by the bounds we have  in (\ref{eq:DisPercProbs}).
The theorem then follows by observing that our assumption that $T\in {\cal A}_{h,\zeta}$  implies
that among the vertices in $\{v_{h/4},\ldots, v_{h}\}$ there are at least $(3/4-\zeta)h$ vertices which 
are mixing roots at their subtree.

Thus, it remains to show the bound in (\ref{eq:DisPercProbs}). In particular, it suffices to show the bound 
re\-gar\-ding the case where the $\tilde{T}_{v_{j+1}}$ has  mixing root, as the other one is 
trivial. For this case assume that $X(v_j)=c, Y(v_j)=q$ for two different $c,q\in [k]$. In 
this situation we have disagreement between $X(v_{j+1}), Y(v_{j+1})$ if either 
$X(v_{j+1})=q$ or $Y(v_{j+1})=c$ or both.  Otherwise, i.e. conditional that 
$X(v_{j+1})\neq q$ and $Y(v_{j+1})\neq c$, there is a coupling such that with probability 1, we have
  $X(v_{j+1})= Y(v_{j+1})$.
Then it becomes apparent that
\begin{eqnarray}\nonumber
\Pr\left[X(v_{j+1}) \neq Y(v_{j+1})| X(v_j)=c, Y(v_j)=q\right] &\leq&  \\ 
&& \hspace{-6cm} \leq\max\left\{
\Pr\left[X(v_{j+1})=q|X(v_{j})=c\right],\Pr\left[Y(v_{j+1})=c|Y(v_{j})=q\right]
 \right\}. \nonumber
\end{eqnarray}
The  result follows almost directly. W.l.o.g. consider the term 
$\Pr\left[X(v_{j+1})=q|X(v_{j})=c\right]$. Clearly there is a $c'\in [k]$
such that
$$
\Pr\left[X(v_{j+1})=q|X(v_{j})=c\right]\leq \Pr\left[X(v_{j+1})=q|X(v_{j})=c, X(v_{j+2})=c'\right].
$$
The above holds because $\Pr\left[X(v_{j+1})=q|X(v_{j})=c \right]$ can be written as a convex
combination of boundaries on $v_{j+2}$. 

We have assumed that $\tilde{T}_{v_{j+1}}$ has  mixing root, while $\sigma_L \in {\cal U}_{w}$.
Then it is elementary to verify that 
$\Pr\left[X(v_{j+1})=q|X(v_{j})=c, X(v_{j+2})=c'\right]\leq 2\Delta^{-\gamma}_+$.
Essentially, this bound follows by using arguments very similar to those for Lemma \ref{lemma:IndBiasCond}.
We omit the derivations. The proposition follows. \hfill $\Box$\\

\section{Proof of Proposition \ref{prop:prob-UXL}}\label{sec:prop:prob-UXL}

So as to show Proposition \ref{prop:prob-UXL} we use the following result.

\begin{proposition}\label{prop:BiasProb}
Let  $\alpha,\gamma,\delta,\Delta_+,\zeta$  be as in the statement of Proposition \ref{prop:prob-UXL}.
Let $k=(1+\alpha){\Delta_+}/{\ln \Delta_+}$.
Consider some tree $H$,  of height $t>0$, which has mixing root.  For  $Z$, a random $k$-colouring of $H$, the following is true 
\begin{eqnarray}\label{eq:ProBiasRandBound}
 \Pr\left[Z_{L_h(H)} \notin {\cal U}(H)\right]\leq 
\exp\left(-\frac18(\Delta_+)^{\frac{t-1}{2}\delta+\frac{7}{4}\frac{\alpha}{1+\alpha}} \right),
\end{eqnarray}
 we remind the reader that  ${\cal U}(H)$ denote the set of all boundary conditions which are not biasing root.
\end{proposition}
The proof of Proposition \ref{prop:BiasProb} appears in Section \ref{sec:prop:BiasProb}.\\

\begin{propositionproof}{\ref{prop:prob-UXL}}
The  proposition follows by using Proposition \ref{prop:BiasProb} and a simple union bound.  
In particular,   let $L=L_h(T)$.  Also, let ${\cal P}$ denote the path that connects $r_T$ and $w\in L_h(T)$
while 
$$
{\cal M}=\left\{v\in {\cal P}: dist(r_T,v)\leq \frac34h, \; \tilde{T}_v \textrm{ has mixing root} \right\}.
$$
Clearly,  $X_L \notin {\cal U}_{w}$ if for some vertex $u\in {\cal M}$,  
it holds that $X(L\cap \tilde{T}_u)\notin {\cal U}(\tilde{T}_u)$, i.e the boundary   
$X(L\cap \tilde{T}_u)$ biases the root of the subtree $\tilde{T}_v$.
 That is,
\begin{eqnarray}
\Pr\left[X(L)\notin {\cal U}_{w}\right]&=&\Pr\left[\bigcup_{u\in {\cal M}} X_{L\cap  \tilde{T}_v} \notin {\cal U}(\tilde{T}_u)\right] 
\leq \sum_{u\in {\cal M}}\Pr\left[ X_{L\cap \tilde{T}_v} \notin{\cal U}(\tilde{T}_u) \right] \qquad \qquad \mbox{[union bound]} \nonumber \\
&\leq&\sum^h_{t= (1/4)h}   \exp\left(-\frac18(\Delta_+)^{\frac{t-1}{2}\delta+\frac{7}{8}\frac{\alpha}{1+\alpha}} \right)
\leq2 \exp\left(-\frac18(\Delta_+)^{\frac{h/4-1}{2}\delta+\frac{7}{8}\frac{\alpha}{1+\alpha}} \right),   \nonumber 
\end{eqnarray}
in the last line, above, we used Proposition \ref{prop:BiasProb}. The proposition follows.
\end{propositionproof}

\section{Proof of Proposition \ref{prop:BiasProb}}\label{sec:prop:BiasProb}

Since we assumed that the tree $H$ has a mixing root, it holds that  $deg(r_H)=s\leq \Delta_+$.
We let $v_1,v_2,\ldots, v_s$ denote the children of $r_H$.   We remind the reader that the set 
 $\mathbb{S}\subseteq \{\tilde{H}_{v_1}, \tilde{H}_{v_2},\ldots, \tilde{H}_{v_s}\}$
contain  only the subtrees whose roots are mixing.

So as to prove Proposition \ref{prop:BiasProb} we need the following result.

\begin{lemma}\label{lemma:FactorRandBoundPro}
Let $X$ be a random $k$-colouring $H$. For $L_i=L_{h-1}(\tilde{H}_{v_i})$, let ${B}_i$ denote the event that in 
$\tilde{H}_{v_i}$, the boundary $X(L_i)$ does not bias $r(\tilde{H}_{v_i})$. For any $\Gamma\subseteq \{1,\ldots, s\}$
it holds that
\begin{eqnarray}
\Pr\left[\cap_{i\in \Gamma} {\cal B}_i\right] = \prod_{i\in \Gamma}\Pr[{\cal B}_i]=\left( \Pr[{\cal B}_i]\right)^{|\Gamma|}. \nonumber
\end{eqnarray}
\end{lemma}
The proof of this lemma is straightforward so we omit it. Essentially, it follows from the fact that a biasing  (resp. non-biasing) 
boundary condition remains biasing (resp. non-biasing) if we repermute the colour classes.
A similar lemma appears in  \cite{TreeNonNaya}.\\

\begin{propositionproof}{\ref{prop:BiasProb}}
The proof is by induction on $t\geq 1$.   The induction basis is  $t=1$.  Then, $H$ is one level tree whose root is of degree at 
most $\Delta_+$.  Let $Y$ denote the number of different colours that do not appear in $X(L_1)$. 
It holds that
\begin{eqnarray}\label{eq:IndHypBiasBound}
 \Pr[X_{L_1(H)} \notin {\cal U}(H)] &\leq & \Pr[Y \leq  \Delta^{\gamma}_+].
\end{eqnarray}
Observe that $\Pr\left[Y\leq \Delta^{\gamma}_+\right]$  is an increasing function of the degree of 
$r(H)$. That is, the larger the degree of $r(H)$ the  more colours are expected to be used to colour 
the leaves of $H$. For this reason, we are going to upper bound the r.h.s. of (\ref{eq:IndHypBiasBound}) 
by assuming that $deg(r_H)=\Delta_+$, i.e. the maximum degree possible for a mixing root. It holds that
\begin{eqnarray}
  \mathbb{E}[Y] &=& (k-1)\left(1-\frac{1}{k-1}\right)^{\Delta_+} 
\geq (k-1)\exp \left(-\frac{\Delta_+}{k-2} \right) \hspace*{2cm}\mbox{[as $1-x\geq e^{\frac{x}{1-x}}$ for $0<x<1/5$]}\nonumber\\
  &\geq &(k-1)\exp \left(-\left(1-\frac{\alpha}{1+\alpha}\right)\ln \Delta_+ - \frac{\ln \Delta_+}{k-2} \right)
 \geq  (\Delta_+)^{\frac{7}{8}\frac{\alpha}{1+\alpha}}.  \label{eq:E[Y]Lower}
\end{eqnarray}
Viewing the $k-1$ colours which are available for the leaves of $H$ as bins and each leaf of $H$ as a ball
which is thrown to a random bin, $Y$ corresponds to the number of empty bins. 
It is a standard result that we can apply Chernoff bounds for bounding the tails of $Y$,  e.g. see \cite{RandAlgBook}.
Then we get that
\begin{eqnarray}
 \Pr\left[Y < (\Delta_+)^{\gamma}\right]&\leq& \Pr\left[Y\leq \mathbb{E}[Y]/2\right]\leq \exp\left(-{\mathbb{E}[Y]}/{8}\right)
 \leq \exp\left(-{(\Delta_+)^{\frac{7}{8}\frac{\alpha}{1+\alpha}}}/{8}\right), \quad
\mbox{[as $\gamma\leq \min\left\{\alpha/2,1/10\right\}$]}\nonumber
\end{eqnarray}
where in the last inequality we use (\ref{eq:E[Y]Lower}).  We have proved the basis of our induction.

Assume, now,  that (\ref{eq:ProBiasRandBound}) is true for every tree of height $t-1$ which has 
mixing root.  It suffices to show that (\ref{eq:ProBiasRandBound}) is true for a tree $H$ of 
height $t$  with a mixing root.  
For such a tree $H$ let $L=L_t(H)$.  Consider also a random $k$-colouring $X$ for this tree.
Let $Z$, denote the number of  subtrees in $\mathbb{S}$ which are biased under the random 
colouring $X_L$,  i.e. the number of trees  $\tilde{H}_{v_i}\in \mathbb{S}$ such that $X(L\cap \tilde{H}_{v_i})$ 
is biasing for $r(\tilde{H}_{v_i})$. 
From Lemma \ref{lemma:IndBiasCond} we have the following
\begin{eqnarray}\label{eq:X(L)ProbBound}
\Pr\left[X_L \notin{\cal U}(H)\right] \leq  \Pr\left[Z>\Delta^{\delta}_+\right].
\end{eqnarray}
Let
\begin{eqnarray}\nonumber
\varrho=max_{\tilde{H}_v\in \mathbb{S}}\left\{\Pr[X(L\cap \tilde{H}_{v}) \notin {\cal U}(\tilde{H}_v) ]\right \},
\end{eqnarray}
where for the subtree $\tilde{H}_v$, the set  ${\cal U}(\tilde{H}_v)$  contains all the boundary conditions (at level $t-1$) 
$\tilde{H}_v$ which do not bias the root of $r(\tilde{H}_v)$.
From Lemma \ref{lemma:FactorRandBoundPro} we conclude that $Z$ is dominated by 
${\cal B}(\Delta_+, \varrho)$,
i.e. the binomial distribution with parameters $\Delta_+$ and $\varrho$.
Due to our assumptions it holds that $\Delta_+^{\delta}\gg \Delta_+\cdot \varrho$.
We have that
\begin{eqnarray}
\Pr\left[Z>\Delta^{\delta}\right] &\leq & \sum^{\Delta_+}_{j=\Delta_+^{\delta}} {\Delta_+ \choose j} \varrho^j\left(1-\varrho \right)^{\Delta_+-j} 
 \leq \Delta_+ {\Delta_+ \choose \Delta_+^{\delta}} \varrho^{\Delta_+^{\delta}}\left(1-\varrho \right)^{\Delta_+-\Delta_+^{\delta}} \nonumber \\ 
&\leq & \frac{\Delta_+}{(\Delta_+^{\delta}/e)^{\Delta_+^{\delta}}}  ( \Delta_+ \varrho)^{\Delta_+^{\delta}} \hspace*{4.35cm} \mbox{[as ${n \choose i}\leq \left( ne/i\right)^i$]} \nonumber \\
&\leq & (\Delta_+ \varrho)^{\Delta_+^{\delta}}\hspace*{5.9cm} \left[ \textrm{as } \frac{\Delta_+}{(\Delta_+^{\delta}/e)^{\Delta_+^{\delta}}}<1 \right] \nonumber \\
&\leq & \left(\Delta_+ \exp\left(-\frac18\Delta_+^{\frac{t-2}{2}\delta+\frac{7}{8}\frac{\alpha}{1+\alpha}}\right)\right)^{\Delta_+^{\delta}} \hspace*{3cm} \mbox{[by the induction hypothesis]}\nonumber \\
&\leq & \left(\exp\left(-\frac18\Delta_+^{\frac{t-3}{2}\delta+\frac{7}{8}\frac{\alpha}{1+\alpha}}\right)\right)^{\Delta_+^{\delta}} \leq 
\exp\left(-\frac18\Delta_+^{\frac{t-1}{2}\delta+\frac{7}{8}\frac{\alpha}{1+\alpha}}\right). \label{eq:NoOfBiasedTail} 
\end{eqnarray}
The proposition follows by plugging (\ref{eq:NoOfBiasedTail}) into (\ref{eq:X(L)ProbBound}).
\end{propositionproof}

\subsection{Proof of Lemma \ref{lemma:IndBiasCond}}\label{sec:lemma:IndBiasCond}

The proof is by induction on the height of the tree $t$.
The case where $t=1$ follows from Definition \ref{def:BiasBound}.

Consider some $t>1$ and assume that the assertion is true for any
tree of height less than $t$. We are going to show that the assertion is true
for trees of height $t$, as well.

Assume that $deg(r_H)=s$ for some integer $s$. Clearly $s\leq \Delta_+$
since we assume that $H$ has a mixing root.  
We let   $v_1,\ldots, v_s$ be  the children of the root. 
Also,  we let $L_i=L\cap \tilde{H}_{v_i}$, where $L=L_t(H)$.
That is $L_i$ denotes the vertices at level $t-1$ of the subtree $\tilde{H}_{v_i}$.

Let $X$ be a random $k$-colouring of $H$ such that $X_L=\sigma_L$
also, for $i=1,\ldots, s$, let $X_i=X(\tilde{H}_{v_i})$. 
A standard  recursive argument yields the following relation: For any $c\in [k]$ it holds that
\begin{eqnarray}
\Pr[X(r_H)=c]&=&\frac{\prod^s_{i=1} \Pr[X_i(v_i)\neq c]  }{\sum_{c'\in [k]}\prod^s_{i=1}Pr[X_i(v_i)\neq c']} \leq 
\frac{1}{\sum_{c'\in [k]}\prod^s_{i=1}\Pr[X_i(v_i)\neq c']}. \label{eq:recX=c}
\end{eqnarray}
We  show that $r(H)$ if  $\sigma_{L_h}$ is non-biasing then  the denominator in (\ref{eq:recX=c}) is sufficiently small.

Let $B\subset [k]$ denote the set of colours $c$ for which there is some $i$ such that 
$\Pr[X_i(v_i)=c]\geq \Delta_+^{-\gamma}$.
It is only $\Delta_+^{\gamma}$ many colours can have increased bias at the root of $\tilde{H}_{v_i}$
since $\sum_{c\in [k]}\Pr[X_i(v_i)=c]=1$.

We have assumed that  there are at most $\Delta_+^{\delta}$ trees $\tilde{H}_{v_i}$ whose root is 
mixing but the boundary biases the colour assignment of the root.
Furthermore,  there are $\Delta_+^{\delta}$ trees $\tilde{H}_{v_i}$ with non-mixing roots. 
That is, there can be at most $2\Delta_+^{\delta}$ trees $\tilde{H}_{v_i}$ whose roots are biased,
those whose root is  biased by the boundary condition and those which have non-mixing root.

Clearly, all the above imply that $|B|\leq 2\Delta_+^{\gamma+\delta}$. Letting $U=[k]\backslash B$,
we  rewrite (\ref{eq:recX=c}) as follows:
\begin{eqnarray}
 \Pr[X(r_H)=c] &\leq & \left (\sum_{c'\in U} \prod^s_{i=1}(1-Pr[X_i(v_i)=c'])\right )^{-1} \nonumber \\
 &\leq & \left (\sum_{c'\in U} \prod^s_{i=1} \exp \left(-\frac{Pr[X_i=c']}{1-Pr[X_i=c']}\right) \right)^{-1}
 \hspace*{2.25cm} \mbox{[as $1-x>e^{x/(1-x)}$ for $0<x<0.1$]}\nonumber \\
 &\leq & \left (|U| \sum_{c'\in U} \frac{1}{|U|} \exp \left(-\sum^s_{i=1}\frac{Pr[X_i(v_i)=c']}{1-Pr[X_i(v_i)=c']}\right) \right)^{-1}\nonumber \\
 &\leq & \left (|U|\prod_{c'\in U}  \exp \left(-\frac{1}{|U|}\sum^s_{i=1}\frac{Pr[X_i(v_i)=c']}{1-Pr[X_i(v_i)=c']}\right) \right)^{-1}\nonumber 
 \; \mbox{[ arithmetic-geometric mean ]}\\
 &\leq & \left (|U|\exp \left(-\frac{1}{|U|}\sum^s_{i=1}\sum_{c\in U}\frac{Pr[X_i(v_i)=c']}{1-Pr[X_i(v_i)=c']}\right)\right)^{-1}\nonumber\\
 &\leq & \left (|U|\exp \left(-\frac{1}{|U|}\sum^s_{i=1}\frac{Pr[X_i(v_i)\in U]}{1-\Delta_+^{-\gamma}} \right)\right)^{-1} 
 \hspace*{1.8cm} \mbox{[as $Pr[X_i(v_i)=c]<\Delta_+^{-\gamma}$ for $c\in U$]} \nonumber \\
 &\leq & \left (|U|\exp \left(-\frac{1}{1-\Delta_+^{-\gamma}}\frac{s}{|U|}  \right)\right)^{-1}.
 \hspace*{3.6cm} \mbox{[as $Pr[X_i\in U]\leq 1$]} \nonumber
 \end{eqnarray}
It is straightforward to show that $|U|\geq k\left(1-\Delta_+^{\frac{\gamma+\delta-1}{2}}\right)\geq
\left (1+\frac{9}{10}\alpha\right)\frac{\Delta_+}{\ln \Delta_+}$, since  $\gamma+\delta<1$.
Also it holds that  $\frac{1}{1-\Delta_+^{-\gamma}}\frac{s}{|U|}\leq \frac{\ln \Delta_+}{1+4\alpha/5}$, since $s\leq \Delta_+$. 
Thus, we get that 
\begin{eqnarray}
\Pr[X=c] \leq \frac{1}{(1+\alpha/2)\frac{\Delta_+}{\ln \Delta_+}\Delta^{-\frac{1}{1+4\alpha/5}}}\leq \Delta_+^{-\frac{3\alpha/5}{1+4\alpha/5}}< \Delta_+^{-\gamma},\nonumber
\end{eqnarray}
as $\gamma= \min\{\alpha/2, 1/10\}$. The lemma follows.

\section{Proof of Proposition \ref{prop:TinAhz}}\label{sec:prop:TinAhz}

For $i=(1-\zeta)h$ we let $Q_{h,i}=\Pr\left[{\cal T}^h_{\xi}\notin {\cal A}_{h, \zeta}\right]$. 
Also, we let $Q^t_{h,i}=\Pr\left[\left. {\cal T}^h_{\xi}\notin {\cal A}_{h, \zeta} \right| {\tt deg}(r(T^h_{\xi}))=t\right]$
Using a simple union bound we get the following: For $t\leq (\Delta_+)^{\delta}$ it holds that
\begin{eqnarray}
 Q^t_{h,i}\leq t\cdot Q_{h-1, i-1}. \label{eq:Q-thi-SmallT}
\end{eqnarray}
Intuitively, the above is implied by the following: If ${\tt deg}(r(T^h_{\xi})) \leq (\Delta_+)^{\delta}$, then, regardless of its children, the root 
$r(T^h_{\xi})$ is mixing.  Conditional that ${\tt deg}(r(T^h_{\xi})) \leq (\Delta_+)^{\delta}$ holds, 
so as to have ${\cal T}^h_{\xi}\notin {\cal A}_{h, \zeta}$, there should be a vertex $v$, child of
$r(T^h_{\xi})$ such that the following is true: The subtree $\tilde{T}_v$ has a path from its root
to its vertices of at level $h-1$ which contain less than $i-1$ mixing vertices.

Using similar arguments,  for $(\Delta_+)^{\delta}\leq t\leq \Delta_+$, we get the following lemma,
whose proof appear in Section  \ref{sec:lemma:Qthi-InterTSolve}.
\begin{lemma}\label{lemma:Qthi-InterTSolve}
For $(\Delta_+)^{\delta}< t\leq \Delta_+$, it holds that
$$
Q^t_{h,i}\leq 2t\left( Q_{h-1,i-1}+  Q_{h-1,i}\cdot \Pr\left [{\cal B}(\Delta_+, q)\geq (\Delta_+)^{\delta}\right] \right).
$$
\end{lemma}

Finally, using a simple union bound we get that for $t>\Delta_+$ it holds that
\begin{eqnarray}
 Q^t_{h,i}\leq t\cdot Q_{h-1, i}.
\label{eq:Q-thi-LargeT}
\end{eqnarray}
The above follows by a line of arguments similar to those we used for (\ref{eq:Q-thi-SmallT}) and
by noting that if ${\tt deg}(r(T^h_{\xi})) \geq \Delta_+$, then the root of $T^h_{\xi}$ is non-mixing.

We are bounding  $Q_{h,i}$ by using (\ref{eq:Q-thi-SmallT}), (\ref{eq:Q-thi-LargeT}) and Lemma 
\ref{lemma:Qthi-InterTSolve}. We have that
\begin{eqnarray}
 Q_{h,i}&=&\sum^n_{t=0}Q^t_{h,i}\xi_t \nonumber\\
&=&Q_{h-1,i-1}\cdot \sum^{(\Delta_+)^{\delta}}_{t=0}t\cdot \xi_t 
+2Q_{h-1,i-1} \cdot \sum^{\Delta_+}_{t=(\Delta_+)^{\delta}+1}t\cdot \xi_t+ \nonumber \\
&&+2Q_{h-1,i} \cdot \Pr\left [{\cal B}(\Delta_+, q)\geq (\Delta_+)^{\delta}\right] \cdot  \sum^{\Delta_+}_{t=(\Delta_+)^{\delta}+1}t\cdot \xi_t+ Q_{h-1,i}\cdot \sum_{t\geq (\Delta_+)+1}t\cdot \xi_t\nonumber\\
 &\leq & 2Q_{h-1,i-1}\sum^{\Delta_+}_{t=0}t\cdot \xi_t+Q_{h-1,i}\left(2\Pr\left [{\cal B}(\Delta_+, q)\geq (\Delta_+)^{\delta}\right]
 \sum^{\Delta_+}_{t=(\Delta_+)^{\delta}}t\cdot \xi_t+\sum_{t\geq (\Delta_+)+1}t\cdot \xi_t\right) \nonumber \\
 &\leq& 
 2d_{\xi}\cdot Q_{h-1,i-1} +  Q_{h-1,i}\left(2d_{\xi}\cdot \Pr\left [{\cal B}(\Delta_+, q)\geq (\Delta_+)^{\delta}\right]+\sum_{t\geq (\Delta_+)+1}t\cdot \xi_t\right).  \label{eq:BasicQih}
\end{eqnarray}

\noindent
The following lemma uses (\ref{eq:BasicQih}) to derive an upper bound on $Q_{h,i}$.

\begin{lemma}\label{lemma:RecTail}
Let $h,\beta, {\cal C}$ be as in the statement of Proposition \ref{prop:TinAhz}.
Also, let  $\lambda \in (0,1)$ and $\theta'>1$ be a fixed numbers such that 
 $\beta(1-\theta')<-1$ and $\lambda\theta'<1$. 
Then for  $i=\lambda h$ and  $Q_{h,i}$ that  satisfy the inequality in 
(\ref{eq:BasicQih}),  
it holds that 
\begin{eqnarray}
Q_{h,i} \leq \exp\left[-(1-\lambda \theta' )\cdot {\cal C}\cdot h  \right].  \label{eq:QhiUpBound}
\end{eqnarray}
\end{lemma}
The proof of Lemma \ref{lemma:RecTail} appears in Section \ref{sec:lemma:RecTail}

The proposition follows by using the above lemma and setting $\lambda=(1-\zeta)$ and $\theta'=\theta$,
where $\zeta$ and $\theta$ are defined in the statement of Proposition \ref{prop:TinAhz}.

\subsection{Proof of Lemma \ref{lemma:Qthi-InterTSolve}}\label{sec:lemma:Qthi-InterTSolve}

Let $q_{h-1}$ be the probability for each child of $r({\cal T}^h_{\xi})$ to be non-mixing.
Conditional that $r({\cal T}^h_{\xi})$ has degree $t$, the number of non-mixing children
of $r({\cal T}^h_{\xi})$ is binomially distributed with parameters, $t$, $q_{h-1}$, i.e.  ${\cal B}(t,q_{h-1})$.
Letting $Q^M_{h,i}=Pr\left[\left. {\cal T}^h_{\xi} \notin {\cal A}_{h, \zeta}\right| r\left(T^h_{\xi}\right) \textrm{ is mixing}\right]$
and  $Q^N_{h,i}=Pr\left[\left. {\cal T}^h_{\xi}\notin {\cal A}_{h, \zeta} \right| r\left(T^h_{\xi}\right) \textrm{ is not mixing}\right]$,
it holds that

\begin{eqnarray}
  Q^t_{h,i} &\leq & \sum_{j=0}^{(\Delta_+)^{\delta}} {t \choose j} q^j_{h-1}(1-q_{h-1})^{t-j}
  \left[(t-j)Q^M_{h-1,i-1}+ jQ^N_{h-1,i-1} \right] + \nonumber\\
 &&+\sum^t_{j=(\Delta_+)^{\delta}+1} {t \choose j} q^j_{h-1}(1-q_{h-1})^{t-j}\left[(t-j)Q^M_{h-1,i}+ jQ^N_{h-1,i} \right]. \nonumber
\end{eqnarray}

\noindent
Using the standard equality that $(t-j){t \choose j}=t{t-1 \choose j}$, we get that
\begin{eqnarray}
  Q^t_{h,i} &\leq & t(1-q_{h-1})Q^M_{h-1,i-1}\sum_{j=0}^{(\Delta_+)^{\delta}} {t-1 \choose j} q^j_{h-1}(1-q_{h-1})^{t-1-j}
  \nonumber\\
  &&+tq_{h-1} Q^N_{h-1,i-1} \sum_{j=1}^{(\Delta_+)^{\delta}} {t-1 \choose j-1} q^{j-1}_{h-1}(1-q_{h-1})^{t-j} \nonumber\\
  &&+ t(1-q_{h-1})Q^M_{h-1,i}\sum^{t-1}_{j=(\Delta_+)^{\delta}+1} {t-1 \choose j} q^j_{h-1}(1-q_{h-1})^{t-1-j}  \nonumber\\
  &&+t q_{h-1} Q^N_{h-1,i} \sum^t_{j=(\Delta_+)^{\delta}+1} {t-1 \choose j-1} q^{j-1}_{h-1}(1-q_{h-1})^{t-j}. \nonumber
\end{eqnarray}
It is not hard to see that for any $h,i $ it holds that $q_{h} Q^N_{h,i}\leq Q_{h,i}$ and 
$(1-q_h)Q^M_{h,i}\leq Q_{h,i}$. Using these two inequalities we get that
\begin{eqnarray}
  Q^t_{h,i}&\leq & tQ_{h-1,i-1}\left(\Pr\left[{\cal B}(t-1, q_{h-1})\leq (\Delta_+)^{\delta}\right] + 
  \Pr\left [{\cal B}(t-1, q_{h-1})\leq (\Delta_+)^{\delta}-1\right]\right)   \nonumber \\
  &&+ tQ_{h-1,i}\left(\Pr\left[{\cal B}(t-1, q_{h-1})\geq  (\Delta_+)^{\delta}+1\right] + \Pr\left [{\cal B}(t-1, q_{h-1})\geq (\Delta_+)^{\delta}\right]\right) \nonumber\\
  &\leq & 2tQ_{h-1,i-1}+  2tQ_{h-1,i}\Pr\left [{\cal B}(t-1, q_{h-1})\geq (\Delta_+)^{\delta}\right]. 
  \label {eq:lemma:Qthi-InterTSolve:1978}
\end{eqnarray}
Note that that $\Pr\left [{\cal B}(t-1, q_{h-1})\geq (\Delta_+)^{\delta}\right]$ is increasing with
$t$. That is,  for $t\leq \Delta_+$ it holds that
\begin{eqnarray}
\Pr\left [{\cal B}(t-1, q_{h-1})\geq (\Delta_+)^{\delta}\right]&\leq &
\Pr\left [{\cal B}(\Delta_+, q_{h-1})\geq (\Delta_+)^{\delta}\right]. \label{eq:lemma:Qthi-InterTSolve:1821}
\end{eqnarray}
At this point we need to observe that the quantity $q$, defined in Definition \ref{def:qDelta+},
is an upper bound for $q_{h}$, for every $h$.   This follows by an inductive argument, i.e. induction
on $h$ the number of levels of ${\cal T}^h_{\xi}$.

Clearly, for $h=0$, the assertion is true. The tree with zero levels consists of only one vertex,  which is  a leaf.
By default the  leaves are mixing vertices, i.e. the probability of a leaf to be non-mixing is zero. 
Since $q\in [0,3/4)$, $q$ is an upper bound for the vertex to be non-mixing.

Given some $h>0$, assume that the assertion is true for ${\cal T}^{h'}_{\xi}$, for any $h'\leq h$ . 
We are going to
show that this is true for $T^{h}_{\xi}$.  Let  $\mathbf{N}$ be the number of
non-mixing children of the root of $T^{h}_{\xi}$. It holds that
\begin{eqnarray}\nonumber
\Pr[r({\cal T}^h_{\xi}) \textrm{ is non-mixing}]\leq 
\Pr[{\tt deg}(r({\cal T}^h_{\xi}))>\Delta_+ ]+\Pr[\mathbf{N}>(\Delta_+)^{\delta}| {\tt deg}(r({\cal T}^h_{\xi}))\leq \Delta_+ ].
\end{eqnarray}
Given  that   ${\tt deg}(r({\cal T}^h_{\xi}))=D$, for some integer $D\geq 0$,  $\mathbf{N}$ is a binomial variable with parameters
$D,q_{h-1}$. Due to our induction hypothesis it holds that $q_{h-1}<q$. Since  we have conditioned that
$D<\Delta_+$,  it is clear that $\mathbf{N}$ is dominated by a binomial variable with parameters $\Delta_+,q$, that is
\begin{eqnarray}
\Pr[r({\cal T}^h_{\xi}) \textrm{ is non-mixing}] &\leq &
\Pr[{\tt deg}(r({\cal T}^h_{\xi}))>\Delta_+ ]+\Pr[{\cal B}(\Delta_+,q)>(\Delta_+)^{\delta}]\nonumber \\
&\leq &
\sum_{i\geq \Delta_+}\xi_i +\Pr[{\cal B}(\Delta_+,q)>(\Delta_+)^{\delta}]\leq q,\nonumber 
\end{eqnarray}
where the last inequality follows from the definition of $q$, i.e. in  Definition \ref{def:qDelta+}.
The above inequality with   (\ref{eq:lemma:Qthi-InterTSolve:1821}) imply that 
$$
\Pr\left [{\cal B}(\Delta_+, q_{h-1})\geq (\Delta_+)^{\delta}\right]\leq 
 \Pr\left [{\cal B}(\Delta_+, q)\geq (\Delta_+)^{\delta}\right], 
$$
as   ${\cal B}(\Delta_+,q_{h-1})$ is stochastically dominated by ${\cal B}(\Delta_+,q)$, since, $q_{h-1}\leq q$, for any $h$.

The lemma follows by plugging the above inequality into (\ref{eq:lemma:Qthi-InterTSolve:1978}).

\subsection{Proof of Lemma \ref{lemma:RecTail}}\label{sec:lemma:RecTail}

We are going to use induction to prove the lemma. First we are going to show that if
(\ref{eq:QhiUpBound})  is true for some $h>1$ then it is also true for $h+1$. 
Let $\lambda=\frac{i}{h}$, $\lambda^-=\frac{i-1}{h-1}$ and $\lambda^+=\frac{i}{h-1}$. 
We rewrite (\ref{eq:BasicQih}) in terms of $\lambda$, $\lambda^+$ and $\lambda^-$ as 
follows:
\begin{eqnarray}\label{eq:BasicRecTail}
 Q_{\{h,\lambda h\}}\leq  2d\cdot Q_{\{h-1,\lambda^- (h-1)\}} +  Q_{\{h-1,\lambda^+(h-1)\}}\left(2d\Pr\left [{\cal B}(\Delta_+, q)\geq (\Delta_+)^{\delta}\right]+\sum_{t\geq (\Delta_+)+1}t\cdot \xi_t\right).
\end{eqnarray}
Using the induction hypothesis and noting that $\lambda^-=\lambda-\frac{1-\lambda}{h-1}$
we have that 
\begin{eqnarray}
 Q_{\{h-1,\lambda^-(h-1)\}}&\leq & \exp\left[ -(1-\theta  \lambda^-)(h-1) {\cal C} \right]\nonumber\\
 &\leq&  \exp\left[ -\left(1-\theta'  \left(\lambda-\frac{1-\lambda}{h-1}\right)\right)(h-1) {\cal C}\right] \nonumber\\
 &\leq&  \exp\left[ -\left(1-\theta' \lambda \right) (h-1){\cal C} \right] \cdot 
 \exp\left[ -\theta' \left({1-\lambda}\right){\cal C} \right] \nonumber \\
 &\leq&  \exp\left[ -\left(1-\theta'  \lambda \right)  h \; {\cal C}  \right] \cdot 
 \exp\left[ \left(1-\theta' \right){\cal C} \right]. \nonumber 
\end{eqnarray}
As far as $Q_{\{h-1,i\}}$ is regarded, we use the fact that $\lambda^+=\lambda+\frac{\lambda}{h-1}$
and we get that 
\begin{eqnarray}
 Q_{\{h-1,\lambda^+\cdot(h-1)\}}&\leq & \exp\left[ -(1-\theta'  \lambda^+) (h-1) {\cal C} \right] \nonumber\\
 &\leq & \exp\left[ -\left(1-\theta'  \lambda-\frac{\theta'\lambda}{h-1}\right)(h-1){\cal C}  \right] \nonumber\\
 &\leq & \exp\left[ -\left(1-\theta'  \lambda \right)(h-1){\cal C}  \right] \cdot 
 \exp\left[ \theta' \lambda{\cal C } \right]\nonumber\\
 &\leq & \exp\left[ -\left(1-\theta'  \lambda  \right)h{\cal C}  \right]
 \exp\left[ {\cal C} \right]\nonumber.
\end{eqnarray}
Substituting the bounds for $Q_{\{h-1,i-1\}},Q_{\{h-1,i\}}$ above into (\ref{eq:BasicRecTail})
we get that
\begin{eqnarray}\label{eq:QhahWithBounds}
 Q_{\{h,\lambda h\}} &\leq & \exp\left[ -\left(1-\theta'  \lambda \right)h {\cal C}  \right] \times \nonumber\\
 &&\times \left ( 2d \cdot  \exp\left[ \left(1-\theta' \right){\cal C} \right]+  \exp\left( {\cal C} \right) \left(2d\Pr\left [{\cal B}(\Delta_+, q)\geq (\Delta_+)^{\delta}\right]+\sum_{t\geq (\Delta_+)+1}t\cdot \xi_t\right)\right ).\nonumber
\end{eqnarray}
From  to our assumption that  $\beta(1-\theta')<-1$ it is direct that 
\begin{eqnarray}
 2d \cdot \exp\left[ \left(1-\theta' \right){\cal C} \right]=2d^{1+\beta(1-\theta')}\leq 1/5.  \nonumber
\end{eqnarray}
Also due to our assumptions about $\Delta_+,\delta$ we get that 
\begin{eqnarray}
  \exp\left( {\cal C} \right) \left(2d\Pr\left [{\cal B}(\Delta_+, q)\geq (\Delta_+)^{\delta}\right]+\sum_{t\geq \Delta_++1}t\cdot \xi_t\right)&\leq& \frac{2}{5}.\nonumber
\end{eqnarray}
Using the two bounds above (\ref{eq:QhahWithBounds}) writes as follows:
\begin{eqnarray}
 Q_{\{h, \lambda h\}} &\leq & \exp\left[ -\left(1-\theta'\cdot \lambda \right)h {\cal C}  \right].\nonumber
\end{eqnarray}
It remains to show the base of the induction, i.e the case $h=1$.
Since the leaves of the trees are, by default, mixing, for any fixed $\lambda \in (0,1)$ and $h=1$ it holds that
\begin{eqnarray}
 Q_{\{h,\lambda \cdot h\}}\leq \Pr[deg(r(T))\geq \Delta_+]=\sum_{t\geq \Delta_+}\xi_t\leq \exp\left[-2{\cal C}\right ]
 \leq \exp\left[ -\left(1-\theta'\cdot \lambda \right) {\cal C}  \right],  \nonumber
\end{eqnarray}
as $\lambda,\theta>0$ while $\lambda\cdot \theta'<1$.
The lemma follows.

\section{Proof of Proposition \ref{prop:reduction}}\label{sec:prop:reduction}

Given some   $\sigma_{L}\in [k]^L$ , we let the variable $Y=Y(\sigma_L)$   be such that
$Y=\mu^{\sigma_L}_{r_T}(c)-1/k$.  Let the colouring of the root $\tau_{r}=c$.
By definition, we have that
\begin{eqnarray}
\mathbb{E}_{\mu^{\tau_{r}}}[Y]&=&\sum_{\sigma_{L}\in [k]^L}\mu^{\tau_{r}}_{L}(\sigma_L)Y(\sigma_L)\nonumber\\
&=&\sum_{\sigma_{L}\in [k]^L}\mu^{\tau_{r}}_{L}(\sigma_L)(\mu^{\sigma_L}(c)-1/k)=\mu^{X(L)}(c)-1/k.\nonumber
\end{eqnarray}
Also, we have that 
\begin{eqnarray}
\mathbb{E}_{\mu^{\tau_{r}}}[Y]&=&
\sum_{\sigma_{L}\in [k]^L}\frac{\mu^{\tau_{r}}_{L}(\sigma_L)}{\mu_{L}(\sigma_L)}( \mu^{\sigma_L}(c)-1/k)
\cdot \mu_{L}(\sigma_L) \nonumber \\
&=& \sum_{\sigma_{L}\in [k]^L}\frac{\mu^{\sigma_L}_{r}(c)}{\mu_{r}(c)}( \mu^{\sigma_L}(c)-1/k)
\cdot \mu_{L}(\sigma_L). \nonumber 
\end{eqnarray}
That is, in order to compute the expectation above we calculate the Randon-Nikodym derivative.
The  derivation in the second line is just an application of Bayes' rule.  Letting $\frac{\mu^{\sigma_L}_{r}(c)}{\mu_{r}(c)}=r(\sigma_L)$
and noting that $\mu_r(c)=1/k$,  it is elementary to verify that 
\begin{eqnarray}\nonumber
k\cdot Y(\sigma_L)+1=r(\sigma_L).
\end{eqnarray}
Using the above equality we get that
\begin{eqnarray}
\mathbb{E}_{\mu^{\tau_{r}}}[Y]&=&k \sum_{\sigma_{L}\in [k]^L}( \mu^{\sigma_L}(c)-1/k)^2\mu(\sigma_L)+
\sum_{\sigma_{L}\in [k]^L}( \mu^{\sigma_L}(c)-1/k)\mu(\sigma_L).
\end{eqnarray}
It is direct to show that $\sum_{\sigma_{L}\in [k]^L}( \mu^{\sigma_L}(c)-1/k)\mu(\sigma_L)=0$.
Thus, we get that
\begin{eqnarray}
\mathbb{E}_{\mu^{\tau_{r}}}[Y]=\mathbb{E}[Y^2]=\mu^{X(L)}(c)-1/k.  \label{eq:updownSqBound}
\end{eqnarray}
where the second expectation is w.r.t. the unconditional Gibbs distribution. Observe that $\mathbb{E}_{\mu^{\tau_{r}}}[Y]\geq 0$.

Using the above equality and Cauchy-Schwarz inequality we get  the following: 
\begin{eqnarray}
\sum_{\sigma(L)\in[k]^L} \mu_L(\sigma_L) \cdot \left |\mu^{\sigma_L}_{r(T)}(c) -1/k\right|
&\leq& \sqrt{\sum_{\sigma(L)\in[k]^L} \mu_L(\sigma_L) \cdot \left |\mu^{\sigma_L}_{r(T)}(c) -1/k\right|^2}\qquad \mbox{[Cauchy-Schwarz]}\nonumber \\
&\leq& \sqrt{\frac{1}{k}\left|\mu^{X_L}_{r(T)}(c)-1/k\right|}.  \hspace*{3.375cm} \mbox{[from (\ref{eq:updownSqBound})]} \label{eq:HalfWay2Reduction}
\end{eqnarray}
Observe that in (\ref{eq:HalfWay2Reduction}) the quantity inside the absolute value is always non-negative
(e.g. from \ref{eq:updownSqBound}).
Also, it holds that 
\begin{eqnarray}
\left|\mu^{X_L}_{r(T)}(c) - 1/k\right| \leq   ||\mu^{X_L}(\cdot) - \mu(\cdot) ||_{r_T}=||\mu^{X_L}(\cdot) - \mu^{Z_L}(\cdot) ||_{r_T}. \label{eq:Reduction2ExtraRandomBoundary}
\end{eqnarray}
where $Z$ is a random $k$-colouring of $T$. The equality, above, holds since the  distributions $\mu_{r_T}$ and $\mu^{Z_L}_{r_T}$ are identical.  
For every $q\in [k]$ let $Z^q$ denote a random colouring of $T$ conditional
that $r(T)$ is coloured $q$.  By the definition of total variation distance we 
get the following:
\begin{eqnarray}
 ||\mu^{X_L}(\cdot) - \mu^{Z_L}(\cdot) ||_{r_T}&=&
\frac12 \sum_{c'\in [k]}\left|\mu^{X_L}_{r_T}(c')-\mu^{Z_L}_{r_T}(c') \right| 
\leq \frac12 \sum_{c'\in [k]}
\left|\mu^{X_L}_{r_T}(c')-\frac{1}{k}\sum_{q\in [k]}\mu^{Z^q_L}_{r_T}(c') \right| \nonumber \\
&\leq&\frac{1}{k}\sum_{q\in [k]}\frac12 \sum_{c'\in [k]}
\left|\mu^{X_L}_{r_T}(c')-\mu^{Z^q_L}_{r_T}(c') \right| \nonumber \\
&\leq&\frac{1}{k}\sum_{q\in [k]}
\left|\left| \mu^{X_L}(\cdot)-\mu^{Z^q_L}(\cdot) \right|\right|. 
\label{eq:Redu2CondRandBound}
\end{eqnarray}

\noindent
Since the r.h.s. of (\ref{eq:Redu2CondRandBound}) is a convex combination,
it follows that 
\begin{eqnarray}
 ||\mu^{X_L}(\cdot) - \mu^{Z_L}(\cdot) ||_{r_T}&\leq&\max_{q\in [k]}
\left\{ \left|\left| \mu^{X_L}(\cdot)-\mu^{Z^q_L}(\cdot) \right|\right|  \right\}. \nonumber
\end{eqnarray}
The proposition follows by combining the above inequality, (\ref{eq:Reduction2ExtraRandomBoundary}) and (\ref{eq:HalfWay2Reduction}).

\section{Proof of Theorem \ref{thrm:ReconThr} - Reconstruction}\label{sec:thrm:ReconThrB}

Consider the following.

\begin{definition}[Freezable Root]\label{def:non-mixing-root}
Consider $\Delta_-$ and $\delta$ as in the statement of Theorem \ref{thrm:ReconThr}.
 For  a tree $T$ of height $t$, its root is freezable if the following holds:
If $t=1$, then $r(T)$ is of degree is at least $\Delta_-$. If $t>1$, $r(T)$ is  
freezable  if and only if  $deg (r_T)\geq  \Delta_-$ and there are at  least 
$\Delta_--(\Delta_-)^{\delta}$ 
many   vertices $v$ children  of $r(T)$ such that $\tilde{T}_v$ has a freezable root.
\end{definition}

\begin{definition}[Freezing Boundary]
Let  $T$ be a tree of height $t$, for some integer $t>0$, and let $L=L_t(T)$.  Let $\sigma$ be a $k$-colourings of $T$,
for some $k>0$.
 Then the  boun\-dary condition $\sigma_L$ {\em freezes}
the colouring $r_T$ if the following holds: There exists $c\in [k]$ such that $\mu^{\sigma_L}_{r_T}(c)=1$.
\end{definition}
That is, a freezing boundary condition forces a unique colouring assignment at the root $T$.

Let ${\cal F}_{h}$ denote the set of trees of height $h$ which have freezable root. 
Since the total variation  distance is always non-negative,   it holds 
that
\begin{eqnarray}
\mathbb{E}||\mu^i-\mu^j||_{L_h}\geq \Pr\left[{\cal T}^h_{\xi}\in {\cal F}_h\right]\cdot 
\mathbb{E}\left [||\mu^i-\mu^j||_{L_h} \left|  {\cal T}^h_{\xi}\in {\cal F}_h \right .\right] \label{eq:NonExpWRTFRZABLE}
\end{eqnarray}
The proof is going to be done in two steps. 
We are going to show that taking $k=(1-\alpha)\Delta_-/\ln\Delta_-$, 
both $\Pr\left[{\cal T}^h_{\xi}\in {\cal F}_h\right]$ and $
\mathbb{E}\left [||\mu^i-\mu^j||_{L_h} \left|  {\cal T}^h_{\xi}\in {\cal F}_h \right .\right] $ are bounded
away from zero, for any $h>0$. In particular we have the following:

\begin{lemma}\label{lemma:FreezRootBound}
Given $\xi,\delta,\Delta_-$ as in Theorem \ref{thrm:ReconThr} the following is true:
It holds that $\Pr\left[T^h_{\xi}\in {\cal F}_h\right]\geq 1-g$, where $g$ is from Definition \ref{def:gDelta-}.
\end{lemma}

\begin{remark}
Given $\xi$ and $\Delta_-$, we choose $g$ to be the smallest number which satisfies (\ref{eq:recon}).
We should note  that the quantity $g$ does not depend on $h$, the height of the tree. 
\end{remark}

\begin{lemmaproof}{\ref{lemma:FreezRootBound}}
We are going to use induction to show that $\Pr\left[T^h_{\xi}\notin {\cal F}_h\right] <g$.
For $h=1$, we  use Definition \ref{def:non-mixing-root}, i.e. 
\begin{eqnarray}
\Pr\left[T^h_{\xi}\notin {\cal F}_h\right]=\Pr[{\tt deg}(r({\cal T}^h_{\xi}))<\Delta_-]=\sum_{i<\Delta_-} \xi_i\leq g, \nonumber
\end{eqnarray}
where the last inequality follows from the definition of the quantity $g$, i.e. from Definition \ref{def:gDelta-}.
Assume  now that $\gamma=\Pr\left[{\cal T}^{h-1}_{\xi}\notin {\cal F}_{h-1}\right]\leq g$ is true for some $h>1$.
We are going to show that it is also true that $\Pr\left[{\cal T}^h_{\xi}\notin {\cal F}_h\right]\leq g$.  
Let the ${\cal Y}_r$ denote the event that $r_T$ has less than  $(\Delta_-)-(\Delta_-)^{\delta}$ children which
$v$ such that $\tilde{T}_v$ does not have a freezable root. It holds that 
\begin{eqnarray}
\Pr\left[T^h_{\xi}\notin {\cal F}_h\right]&\leq  &
\Pr\left[ {\tt deg}(r({\cal T}^h_{\xi})) <\Delta^-\right]+\Pr\left[ {\tt deg}(r({\cal T}^h_{\xi})) \geq \Delta^-\right]
\Pr[{\cal Y}_r|{\tt deg}(r({\cal T}^h_{\xi}))\geq \Delta^-] \nonumber \\
&\leq &\sum_{i<\Delta_-}\xi_i+ \sum_{i\geq \Delta_-} \Pr[{\cal Y}_r, {\tt deg}(r({\cal T}^h_{\xi}))=i] \nonumber \\
&\leq &\sum_{i<\Delta_-}\xi_i+ \sum_{i\geq \Delta_-} \xi_i \Pr\left[{\cal B}(i,1-\gamma)<(\Delta_-)-(\Delta_-)^{\delta}\right]   \nonumber\\
&\leq &\sum_{i<\Delta_-}\xi_i+ \sum_{i\geq \Delta_-} \xi_i \Pr\left[{\cal B}(i,1-g)<(\Delta_-)-(\Delta_-)^{\delta}\right] \leq g .
\qquad \mbox{[by Definition \ref{def:gDelta-}]} \nonumber
\end{eqnarray}
The lemma follows.
\end{lemmaproof}

\begin{lemma}
Let $\alpha, \delta, \Delta_-$ be as in Theorem \ref{thrm:ReconThr}.
For  $k=(1+\alpha)\Delta_-/\ln \Delta_-$ it holds that
\begin{eqnarray}
\mathbb{E}\left [||\mu^i-\mu^j||_{L_h} \left|  {\cal T}^h_{\xi}\in {\cal F}_h \right .\right]\geq \left(1-\frac{2}{\log k}\right). \nonumber
\end{eqnarray}
\end{lemma}
\begin{proof}
The lemma will follow by  assuming  any instance of the trees in ${\cal F}_h$, i.e. 
we consider a fixed tree $T\in {\cal F}_h$.
We let $\mathbf{F}$ denote the set of these vertices $v$ children of  $r(T)$ such that $\tilde{T}_v$ has a freezable root.  
Since we have assumed that $T\in {\cal F}_h$ it holds that $|\mathbf{F}|\geq \Delta_--(\Delta_-)^{\delta}$.

Take a random colouring of $T$. W.l.o.g. assume that the root is coloured  with colour $c$. 
This means that each of the children of the root has a colour which is distributed uniformly 
at random in $[k]\backslash\{c\}$ and  each of the colour assignments is independent of the other.  
So as the colour assignment of the root to be frozen, it suffices to have the 
following: For every colour 
$q\in [k]\backslash\{c\}$ there should be at least one child in $\mathbf{F}$ which is assigned 
$q$ and its  colouring is frozen.
Clearly, examining only the children of the $r(T)$ which are in $\mathbf{F}$ will yield a lower
bound for the probability that we have a frozen colouring at $r(T)$.
Let $P_h$ denote the probability that the root of $T$ is frozen.
For the Gibbs distribution of the tree $T$ then it holds that 
$$
||\mu^i-\mu^j||_{L_h} \geq P_h.
$$
Also, since the tree $T$ is chosen arbitrarily from ${\cal F}_h$, we get that
$P_h$ is a lower bound for the expectation $\mathbb{E}\left [||\mu^i-\mu^j||_{L_h} \left|  {\cal T}^h_{\xi}\in {\cal F}_h \right .\right]$, 
too.   The lemma follows by bounding appropriately $P_h$.  

At this point, we can derive the bound by working, essentially, as in  \cite{InfFlowTrees,SemerjianFreez,SlyRecon}.
For the sake of completeness in what follows we present the steps for bounding $P_h$.

Letting $w_q$ denote the number of occurrences of the colour $q$ between
the vertices in $\mathbb{F}$ we have that
\begin{eqnarray}
P_h=\mathbb{E}\left[ \prod_{q\in [k]\backslash\{c\}}\left(1-(1-P_{h-1})^{w_q}\right)\right],
\end{eqnarray}
where the expectation is w.r.t. the random variables $w_q$. Clearly the variables
$w_q$ for different $q$ follow the multinomial distribution. E.g.  the should sum to
$|\mathbf{F}|$. Clearly the random variables are correlated with each other.

Consider a set of $k-1$ independent random variables $\tilde{w}_q$ for every $q\in [k]\backslash\{c\}$.
Each $\tilde{w}_q$ follows a Poisson distribution with parameter $D=\frac{|\mathbf{F}|}{k-1}\left(1-\frac{1}{\log k} \right)$.
It is elementary  to show that conditional that $\sum_{q\in [k]\backslash\{q\}}\tilde{w}_q\leq |\mathbf{F}|$
there is a coupling of $(w_1,\ldots,w_{k-1})$ and $(\tilde{w}_1,\ldots,{w}_{k-1})$ such that
for every $q$ it holds that $w_q\geq \tilde{w}_q$,  (e.g. see Lemma 4 in \cite{SlyRecon} ).
Then clearly we get that
\begin{eqnarray}
P_h &\geq& \mathbb{E}\left[ \prod_{q\in [k]\backslash\{c\}}\left(1-(1-P_{h-1})^{\tilde{w}_q}\right)\right]-\Pr\left[\sum_{q\in [k]\backslash\{c\}}\tilde{w}_q>|\mathbf{F}|\right]\nonumber\\
&\geq& \prod_{q\in [k]\backslash\{c\}} E\left[ \left(1-(1-P_{h-1})^{\tilde{w}_q}\right)\right]-\Pr\left[\sum_{q\in [k]\backslash\{c\}}\tilde{w}_q>|\mathbf{F}|\right]\nonumber\\
&\geq&  \left [1-\exp(P_{h-1}D)\right]^{k-1}-\Pr\left[\sum_{q\in [k]\backslash\{c\}}\tilde{w}_q>|\mathbf{F}|\right],\nonumber
\end{eqnarray}
in the second inequality we use the fact that $\tilde{w}_q$s are independent with each other. 
It holds that $\sum_{q\in [k]\backslash\{c\}}\tilde{w}_q$ is distributed as in Po$\left(|\mathbf{F}|\left(1-1/\log k\right)\right)$.
Thus, it holds that $s=\Pr\left[\sum_{q\in [k]\backslash\{c\}}\tilde{w}_q>|\mathbf{F}|\right] \leq 1/k^2$.

Let $f(x)=(1-\exp\left(xD\right))^{k-1}-s$. Then it is direct to verify that $f(1-\frac{1}{\log k})>1-\frac{1}{\log k}$.
Since  $P_0=1$ and $f(x)$ is increasing function we get that $P_h>1-\frac{1}{\log k}$, for any $h\geq 0$.
\end{proof}

\section{Proof of Theorem \ref{thrm:Threshold}}\label{sec:thrm:Threshold}

We will show the theorem by using  Theorem \ref{thrm:ReconThr}.

Let  $\xi$ be a distribution on the  non-negative integers such
that it is well-concentrated. Also let $d_{\xi}$ be the expected value of $\xi$.
We assume that  $d_{\xi}$ is  sufficiently large.

The theorem  follows by showing that for any fixed $\alpha>0$, for  $k_1=(1+\alpha)d_{\xi}/\ln d_{\xi}$
and $k_2=(1-\alpha)d_{\xi}/\ln d_{\xi}$ the following is true: There exist  appropriate numbers  
$\gamma_1=\gamma_1(\alpha)>0$ and $\gamma_2=\gamma_2(\alpha)>0$ such that 
$d_{\xi}\leq \Delta_+\leq (1+\gamma_1)d_{\xi}$ also 
$d_{\xi}\geq \Delta_-\geq (1-\gamma_2)d_{\xi}$, where $\Delta_+$ and $\Delta_-$
are chosen as specified by Theorem \ref{thrm:ReconThr}.
Furthermore it holds that 
$k_1\geq (1+\alpha/2)\Delta_+/\ln \Delta_+$ and 
$k_2\leq (1-\alpha/2)\Delta_-/\ln \Delta_-$.

Consider, first,   the  quantity $\Delta_+$.  
We choose $\gamma_1$ to be the largest number  such that 
$(1+\alpha)d_{\xi}/\ln d_{\xi}\geq (1+\alpha/2)\rho/\ln \rho$, where $\rho=(1+\gamma_1)d_{\xi}$.
We choose $\gamma_1$ to be independent of $d_{\xi}$. This means that for a given
$\alpha$ and $\gamma_1$, the inequality 
$(1+\alpha)d_{\xi}/\ln d_{\xi}\geq (1+\alpha/2)\rho/\ln \rho$ holds for sufficiently large $d_{\xi}$.

It suffices to show that  $\Delta_+$, chosen as specified in  Theorem \ref{thrm:ReconThr},
 is such that $d_{\xi}\leq \Delta_+\leq (1+\gamma_1)d_{\xi}$.
 Note that the parameter $\delta$ we use for  $\Delta_+$ is such that  $\delta=\min\{ \alpha/4, 1/10\}$.
  
Since $\xi$ is well concentrated,  for any $x\geq (1+\gamma_1)d_{\xi}$ it holds that
\begin{eqnarray}\label{eq:thrm:Threshold:WellConc}
\sum_{i\geq x}\xi_i \leq x^{-c},
\end{eqnarray}
where  $c>0$ is sufficiently large number.
Choosing $q=2d^{-c}_{\xi}$ it is direct to verify that the condition (\ref{eq:non-reconA})
is trivially satisfied by choosing   $\Delta_+\leq (1+\gamma_1)d_{\xi}$. This follows by 
using the inequality in (\ref{eq:thrm:Threshold:WellConc}), i.e. that $\xi$ is well concentrated 
and the Chernoff bounds  for $\Pr[{\cal B}(\Delta_+,q)\geq \Delta^{\delta}_+]$.


The leftmost conditions in (\ref{eq:non-reconB}) is also satisfied for $\Delta_+\leq (1+\gamma_1)d_{\xi}$ and 
sufficiently large $c>0$.
I.e. it holds that
$$
\sum_{t>(1+\gamma_1)d_{\xi}}t\cdot \xi_t\leq \sum_{t>(1+\gamma_1)d_{\xi}}t\cdot t^{-c}\leq 2[(1+\gamma_1)d_{\xi}]^{-(c-1)}.
$$
The second condition in (\ref{eq:non-reconB}) is trivially satisfied, as we describe above. 

Consider now the case of $\Delta_-$.  We work in a very similar way as for the case of $\Delta_+$.
We choose $\gamma_2$ to be the largest number  such that 
$(1-\alpha)d_{\xi}/\ln d_{\xi}\leq (1-\alpha/2)\rho/\ln \rho$, where $\rho=(1-\gamma_2)d_{\xi}$.
We choose $\gamma_2$ to be independent of $d_{\xi}$, in the same manner as we chose
$\gamma_1$, for $\Delta_+$.

It suffices to show that  $\Delta_-$, chosen as specified in  Theorem \ref{thrm:ReconThr},
 is such that $d_{\xi}\geq \Delta_-\leq (1-\gamma_2)d_{\xi}$.
 Note that the parameter $\delta$ we use for  $\Delta_-$ is such that  $\delta=\min\{ \alpha/4, 1/10\}$.

Our assumption that $\xi$ is well concentrated, implies that 
\begin{eqnarray}\label{eq:thrm:Threshold:WellConcB}
\sum_{i\leq (1-\gamma_2)d_{\xi}}\xi_i\leq d^{-c}_{\xi}.
\end{eqnarray}
Setting   $d_{\xi}\geq \Delta_- \geq (1-\gamma_2)d_{\xi}$ and   $g=2d^{-c}_{\xi}$, where $c$ is the same as above,
it suffices to show that  the constraint  (\ref{eq:recon}), in Definition \ref{def:gDelta-}, is satisfied.
In particular,  in the light of  (\ref{eq:thrm:Threshold:WellConc}), it suffices to show that for our choice of $g$ and $\Delta_-$,
the rightmost sum in (\ref{eq:recon}) is sufficiently small.

It holds that $g\cdot \Delta_- < d^{-c/2}_{\xi}\ll (\Delta_-)^{-1+\delta}$.
This implies that for any $i\geq \Delta_-$ we have that
$$\Pr\left[{\cal B}(i,1-g)<(\Delta_-)-(\Delta_-)^{\delta}\right]<\Pr\left[{\cal B}(\Delta_-,1-g)<(\Delta_-)-(\Delta_-)^{\delta}\right],$$
as $\Delta_--\Delta_-^{\delta}<i\cdot g$ for all $i\geq \Delta_-$.
 Thus,  it holds that
\begin{eqnarray}
\sum_{i\geq \Delta_-}\xi_i\Pr\left[{\cal B}(i,1-g)<(\Delta_-)-(\Delta_-)^{\delta}\right]&\leq &
\Pr\left[{\cal B}(\Delta_-,1-g)<(\Delta_-)-(\Delta_-)^{\delta}\right] \sum_{i\geq \Delta_-}\xi_i\nonumber \\
&\leq& \Pr\left[{\cal B}(\Delta_-,1-g)<(\Delta_-)-(\Delta_-)^{\delta}\right]\nonumber\\
&=& \Pr\left[{\cal B}(\Delta_-,g)> (\Delta_-)^{\delta}\right] \leq \exp\left(-\Delta^{\delta}\right) \nonumber.
\end{eqnarray}
The inequality in the second line follows from the fact that $ \sum_{i\geq \Delta_-}\xi_i\leq 1$.
The last inequality follows from a direct application of Chernoff bounds, i.e. Corollary 2.4 in \cite{janson}.
Using the above bounds, it is trivial to show for our choice of $g$ and $\Delta_-$  (\ref{eq:recon}) is true.

 The theorem follows.
\\  \\ \\ 

\noindent
{\bf Acknowledgement.}
The author of this work would like to thank Guilhem Semerjian for our communication and
the discussion on the problem.
Also, the author would like   to thank Amin Coja-Oghlan for the discussions, his comments
 and the suggestions for improving the content of this work.

\end{document}